\documentclass[10pt,twoside]{article}

\usepackage[utf8]{inputenc}
\usepackage[english]{babel}
\usepackage[english, num]{isodate}
\usepackage{amsfonts,amsmath,amssymb,bbm}
\usepackage{amsthm}
\usepackage{nicefrac}
\usepackage{pifont}
\usepackage{yfonts}
\usepackage{geometry} 
\usepackage{graphicx}
\usepackage{booktabs}  
\usepackage{array} 
\usepackage{paralist} 
\usepackage{verbatim} 
\usepackage{subfig} 
\usepackage[nottoc,notlof,notlot]{tocbibind} 
\usepackage[titles,subfigure]{tocloft}

\newcommand{\gs}{{\mathrm{gs}}}
\newcommand{\hf}{{\mathrm{hf}}}

\newcommand{\SD}{{\mathcal{SD}}}
\newcommand{\Ex}{{\mathrm{Ex}}}
\newcommand{\eps}{{\varepsilon}}        
\newcommand{\vphi}{{\varphi}}           

\newcommand{\cA}{\mathcal{A}}
\newcommand{\cB}{\mathcal{B}}

\newcommand{\cD}{\mathcal{D}}

\newcommand{\cF}{\mathcal{F}}

\newcommand{\cL}{\mathcal{L}}         
\newcommand{\cO}{\mathcal{O}}         
\newcommand{\cP}{\mathcal{P}}         

\newcommand{\fh}{\mathfrak{h}}

\newcommand{\uR}{{\underline R}}        

\newcommand{\uZ}{{\underline Z}}

\def\e{{c}^{*}}
\def\v{{c}}

\newcommand{\TRh}[1]{\mathrm{tr}_1\left\{#1\right\}}
\newcommand{\TRhh}[1]{\mathrm{tr}_2\left\{#1\right\}}
\newcommand{\TRF}[1]{\mathrm{tr}_{\mathcal{F}}\left\{#1\right\}}

\newcommand{\vv}{\varphi}

\def\TAI{Documenta Mathematica}
\def\rec{March 14, 2012} 

\def\YEAR{\year}\newcount\VOL\VOL=\YEAR\advance\VOL by-1995
\def\firstpage{1}\def\lastpage{5}
\def\revised{}
\def\communicated{}

\makeatletter
\def\magnification{\afterassignment\m@g\count@}
\def\m@g{\mag=\count@\hsize6.5truein\vsize8.9truein\dimen\footins8truein}
\makeatother

\font\eightrm=cmr8
\font\caps=cmcsc10                    
\font\Caps=cmcsc10 scaled \magstep1   

\makeatletter
\@specialpagefalse\headheight=8.5pt
\def\DocMath{}
\renewcommand{\@evenhead}{%
    \ifnum\thepage>\lastpage\rlap{\thepage}\hfill%
    \else\rlap{\thepage}\slshape\leftmark\hfill{\caps\SAuthor}\hfill\fi}%
\renewcommand{\@oddhead}{%
    \ifnum\thepage=\firstpage{\DocMath\hfill\llap{\thepage}}%
    \else{\slshape\rightmark}\hfill{\caps\STitle}\hfill\llap{\thepage}\fi}%
\makeatother

\def\TSkip{\bigskip}
\newbox\TheTitle{\obeylines\gdef\GetTitle #1
\ShortTitle  #2
\SubTitle    #3
\Author      #4
\ShortAuthor #5
\EndTitle
{\setbox\TheTitle=\vbox{\baselineskip=20pt\let\par=\cr\obeylines%
\halign{\centerline{\Caps##}\cr\noalign{\medskip}\cr#1\cr}}%
	\copy\TheTitle\TSkip\TSkip%
\def\next{#2}\ifx\next\empty\gdef\STitle{#1}\else\gdef\STitle{#2}\fi%
\def\next{#3}\ifx\next\empty%
    \else\setbox\TheTitle=\vbox{\baselineskip=20pt\let\par=\cr\obeylines%
    \halign{\centerline{\caps##} #3\cr}}\copy\TheTitle\TSkip\TSkip\fi%
\centerline{\caps #4}\TSkip\TSkip%
\def\next{#5}\ifx\next\empty\gdef\SAuthor{#4}\else\gdef\SAuthor{#5}\fi%
\ifx\TAI\empty\TSkip%
    \else\centerline{\eightrm To appear in \TAI}\TSkip\fi%
\ifx\rec\empty\relax
    \else\centerline{\eightrm Received: \rec}\fi%
\ifx\revised\empty\TSkip%
    \else\centerline{\eightrm Revised: \revised}\TSkip\fi%
\ifx\communicated\empty\relax
    \else\centerline{\eightrm Communicated by \communicated}\fi\TSkip\TSkip%
\catcode'015=5}}\def\Title{\obeylines\GetTitle}
\def\Abstract{\begingroup\narrower
    \parskip=\medskipamount\parindent=0pt{\caps Abstract. }}
\def\EndAbstract{\par\endgroup\TSkip}

\long\def\MSC#1\EndMSC{\def\arg{#1}\ifx\arg\empty\relax\else
     {\par\narrower\noindent%
     2000 Mathematics Subject Classification: #1\par}\fi}

\long\def\KEY#1\EndKEY{\def\arg{#1}\ifx\arg\empty\relax\else
	{\par\narrower\noindent Keywords and Phrases: #1\par}\fi\TSkip}

\newbox\TheAdd\def\Addresses{\vfill\copy\TheAdd\vfill
    \ifodd\number\lastpage\vfill\eject\phantom{.}\vfill\eject\fi}
{\obeylines\gdef\GetAddress #1
\Address #2
\Address #3
\Address #4
\EndAddress
{\def\xs{4.3truecm}\parindent=0pt
\setbox0=\vtop{{\obeylines\hsize=\xs#1\par}}\def\next{#2}
\ifx\next\empty 
     \setbox\TheAdd=\hbox to\hsize{\hfill\copy0\hfill}
\else\setbox1=\vtop{{\obeylines\hsize=\xs#2\par}}\def\next{#3}
\ifx\next\empty 
     \setbox\TheAdd=\hbox to\hsize{\hfill\copy0\hfill\copy1\hfill}
\else\setbox2=\vtop{{\obeylines\hsize=\xs#3\par}}\def\next{#4}
\ifx\next\empty\ 
     \setbox\TheAdd=\vtop{\hbox to\hsize{\hfill\copy0\hfill\copy1\hfill}
                \vskip20pt\hbox to\hsize{\hfill\copy2\hfill}}
\else\setbox3=\vtop{{\obeylines\hsize=\xs#4\par}}
     \setbox\TheAdd=\vtop{\hbox to\hsize{\hfill\copy0\hfill\copy1\hfill}
	        \vskip20pt\hbox to\hsize{\hfill\copy2\hfill\copy3\hfill}}
\fi\fi\fi\catcode'015=5}}\gdef\Address{\obeylines\GetAddress}

\begin{document}
\theoremstyle{plain}
\newtheorem{thm}{Theorem}[section]
\newtheorem{lem}[thm]{Lemma}
\newtheorem{cor}[thm]{Corollary}
\theoremstyle{definition}
\newtheorem{defn}[thm]{Definition}
\theoremstyle{remark}
\newtheorem{rem}[thm]{Remark}
\bibliographystyle{plain}

\Title
			Fermion Correlation Inequalities
			derived from G- and P-Conditions                                     
\ShortTitle
			Correlation Inequalities from G and P
\SubTitle
			
\Author
		Volker Bach, Hans Konrad Knörr, Edmund Menge
\ShortAuthor
		   Bach, Knörr, Menge
\EndTitle

\Abstract
It is shown in this paper that the G-Condition and the P-Condition
from representability imply the fermion correlation estimate from
\cite{VBA} which, in turn, is known to yield a nontrivial bound on the
accuracy of the Hartree--Fock approximation for large Coulomb systems.
\EndAbstract

\MSC
Primary: 81V70; Secondary: 81V45 and 81V55
\EndMSC

\KEY
Fermion correlation estimate, representability conditions
\EndKEY

\Address
	Volker Bach
	v.bach@tu-bs.de

\Address
	Hans Konrad Knörr
		h-k.knörr@tu-bs.de

\Address
	Edmund Menge
	e.menge@tu-bs.de
	\phantom{.}
	Institut für Analysis und Algebra
	TU Braunschweig
	Pockelsstraße 14,
	38118 Braunschweig
\Address
\EndAddress


\section{Introduction}

The dynamics of $N$ electrons in an atom ($K=1$) or molecule ($K \geq
2$) with $K$ nuclei of charges $\uZ := (Z_1, Z_2, \ldots, Z_K)$, fixed
at positions $\uR := (R_1, R_2, \ldots, R_K)$, is generated by the
Hamiltonian
\begin{align}
H^{(N)}(\uZ,\uR)\ :=\ 
\sum\limits_{n=1}^{N} 
\Bigg(-\Delta_{x_n} -\sum\limits_{j=1}^{K}\frac{Z_j}{\left|x_n-R_j\right|}\Bigg)
+\sum\limits_{1\leq n<m \leq N} \frac{1}{\left|x_n-x_m\right|} \label{Ham1}
\end{align}
to lowest order in the Born--Oppenheimer
approximation. $H^{(N)}(\uZ,\uR)\equiv H^{(N)}$ is a semibounded,
self-adjoint operator defined on a suitable dense domain $\cD^{(N)}$ in
the Hilbert space $\cF_{\mathrm{f}}^{(N)}[\fh]$ of antisymmetric $N$-electron wave
functions, cf.~\eqref{Hamiltonian} below. 

Basic quantities of interest are the ground state energy
\begin{align*} 
E_\gs^{(N)}(\uZ,\uR) \ := \ \inf\left\{\sigma\left\{H^{(N)}(\uZ,\uR)\right\}\right\} ,
\end{align*}
whose variational characterization
\begin{align} \label{eq-1.3}
E_\gs^{(N)}(\uZ,\uR)=
\inf \left\{ \left< \Psi^{(N)}\Big|\,H^{(N)}\Psi^{(N)} \right> \Big|\  
\Psi^{(N)} \in \cD^{(N)} \cap \cF_{\mathrm{f}}^{(N)}[\fh],\ \|\Psi^{(N)}\| = 1 \right\}
\end{align}
is given by the Rayleigh--Ritz principle, and corresponding ground
states $\Psi_\gs^{(N)}$, i.e., normalized solutions of the stationary
Schr\"odinger equation
\begin{align*} 
H^{(N)}(\uZ,\uR) \Psi_\gs^{(N)} 
\ = \
E_\gs^{(N)}(\uZ,\uR) \Psi_\gs^{(N)} .
\end{align*}
The Hartree--Fock (HF) variational principle is an important method to
obtain approximations to both, the ground state energy and ground
states. The HF energy $E_\hf^{(N)}(\uZ,\uR)$ is defined by restricting
the variation in \eqref{eq-1.3} to $\SD^{(N)}[\fh]$,
\begin{align} \label{eq-1.5}
\nonumber E_\hf^{(N)}&(\uZ,\uR)\\
&=\inf\left\{ \left< \Phi^{(N)}\Big|\,H^{(N)}\Phi^{(N)} \right> \Big|\  
\Phi^{(N)} \in \cD^{(N)} \cap \SD^{(N)}[\fh],\ \|\Phi^{(N)}\| = 1 \right\},
\end{align}
where $\SD^{(N)}[\fh] \subseteq \cF_{\mathrm{f}}^{(N)}[\fh]$ denotes the set of
Slater determinants, i.e., the set of all antisymmetrized
product vectors $\vphi_1 \wedge \cdots \wedge \vphi_N$.
Since the variation in \eqref{eq-1.5}, compared to \eqref{eq-1.3}, is
restricted, we clearly have 
\begin{align*} 
E_\hf^{(N)}(\uZ,\uR) 
\ \geq \
E_\gs^{(N)}(\uZ,\uR). 
\end{align*}
A lower bound to the ground state energy by the HF energy minus an
error which is small in the large-$Z$ limit was obtained by one of us
in \cite{VBA,VB2}. In the case of a neutral atom, i.e., $N =
Z := Z_1$ and $R_1 =0$, the resulting estimate was
\begin{align} \label{eq-1.7}
E_\gs^{(Z)}(Z) 
\ \geq \
E_\hf^{(Z)}(Z) - \cO\big( Z^{(5/3) - \eps}\big) ,
\end{align}
for some $\eps >0$. The error term $\cO( Z^{(5/3) - \eps})$ is
small compared to all three contributions to $E_\hf^{(Z)}(Z)$, namely,
the kinetic, the classical electrostatic, and the exchange energy
which are at least of size $c Z^{5/3}$ in magnitude for some
constant $c >0$.

A key inequality derived in \cite{VBA} that
eventually lead to \eqref{eq-1.7} is the fermion correlation estimate
\begin{align} \label{eq-1.8}
\TRhh{(X \otimes X) \Gamma^{\left(\mathrm{T}\right)}}\ \geq\ 
-\, \TRh{ X \gamma}\min\left\{ 1;\ \mathrm{const}\cdot\sqrt{\TRh{X \left(\gamma - \gamma^2\right)}}\right\}
\end{align}

where $X = X^* = X^2$ is an orthogonal projection, 
$\Gamma^{\left(\mathrm{T}\right)} := \Gamma - (1 - \Ex)(\gamma \otimes \gamma)$, $\Gamma
\equiv \Gamma_{\Phi^{(N)}}$ is the two-particle and $\gamma \equiv
\gamma_{\Phi^{(N)}}$ the one-particle density matrix of a normalized
$N$-electron state $\Phi^{(N)} \in \cF_{\mathrm{f}}^{(N)}[\fh]$.
 
The purpose of the present paper is to give an alternative derivation
of \eqref{eq-1.8} by using ideas originating from the theory of
$N$-representability. More precisely, we show that \eqref{eq-1.8}
follows already from the G-Condition and the P-Condition
specified by Garrod and Percus \cite{CJP} and Coleman \cite{AJC}. 

Observing that the Rayleigh--Ritz principle \eqref{eq-1.3} can be
rewritten as a variation over all $N$-representable two-particle
density matrices $\Gamma$, we consequently obtain \eqref{eq-1.7} from
relaxing the requirement of $N$-representability of $\Gamma$ to merely
requiring $\Gamma$ to fulfill the G-Condition and the P-Condition:

\begin{thm}
The G-Condition and the P-Condition imply \eqref{eq-1.8}.
\end{thm}

We note that (\ref{eq-1.8}) was also derived by Graf and Solovej in \cite{SVG} by a different
method that, in retrospective, resembles the application of Garrod
and Percus' G-Condition. In fact, one part of the derivation in \cite{SVG}
follows already from the G-Condition. A main difference to using representability methods, however, lies in the use of
operator inequalities in \cite{SVG} which are necessarily formulated on
the $N$-particle Hilbert space, as opposed to the one- or two-particle Hilbert spaces in the presented work. \\
\ \\
In future work we plan to sharpen this result by making additional use
of Erdahl's $\text{T}_1$- and $\text{T}_2$-Conditions \cite{RME,RME2}
which have recently lead to very good numerical results in quantum
chemistry computations \cite{CML, MAZ, UBS}, as well as Coleman's Q-Condition which was also given in \cite{AJC}
but is not necessary for the derivation of our present result. Furthermore, 
similar representability conditions also exist for bosons \cite{MAZ}. There we like to adress the question whether
analogous results can also be obtained.
\ \\
\ \\
{\sc{Acknowledements.}}
We would like to thank Gero Friesecke, Peter M\"uller and Heinz Siedentop for 
fruitful remarks and discussions. H. K. K. and E. M. were partially supported by the
MPGC Mainz and the ESI.


\section{Density Matrices and Reduced Density Matrices}


\subsection{Fock space, Creation and Annihilation Operators}
Let $\fh$ be a separable complex Hilbert space which we henceforth refer to as the one-particle Hilbert space. 
The fermion Fock space $\cF\equiv\cF_{\mathrm{f}}[\fh]$ is defined to be the orthogonal sum
\begin{align*}
 \cF_{\mathrm{f}}[\fh]:=\bigoplus_{N=0}^{\infty}\cF_{\mathrm{f}}^{(N)}[\fh],
\end{align*}
where 
\begin{align*}
 \cF_{\mathrm{f}}^{(N)}[\fh]:=\stackrel{N}{\bigwedge}\fh:=
	\cA_{N}\Big(\stackrel{N}{\bigotimes}\fh\Big)
\end{align*}
is the antisymmetric tensor product of $N$ copies of $\fh$, for $N\geq 1$, and
$\cF_{\mathrm{f}}^{(0)}[\fh]:=\mathbbm{C}\cdot\Omega$, with $\Omega$ being the normalized 
vacuum vector.
Here, $\cA_{N}$ is the orthogonal 
projection from $\stackrel{N}{\bigotimes}\fh$ onto $\stackrel{N}{\bigwedge}\fh$ uniquely defined by 
\begin{align*}
\nonumber \cA_{N}\left(\varphi_1\otimes\cdots\otimes\varphi_N\right)&:=
    \frac{1}{N!}\sum\limits_{\pi\in\mathcal{S}_N}\left(-1\right)^{\pi}\varphi_{\pi(1)}
    \otimes\cdots\otimes\varphi_{\pi(N)}\\
    &=:\frac{1}{\sqrt{N!}}\varphi_1\wedge\cdots\wedge\varphi_N,
\end{align*}
for $\varphi_1,\dots,\varphi_N\in\fh$. It is convenient to introduce creation operators 
$\e(f)\in\cB(\cF)$ for any $f\in\fh$ by 
\begin{align}
 \e(f)\Omega:=f, \label{DefCre1}
\end{align}
\begin{align}
 \e(f)\left(\varphi_1\wedge\cdots\wedge\varphi_N\right):=f\wedge\varphi_1\wedge\cdots\wedge\varphi_N 
												  \label{DefCre2}
\end{align}
for $\varphi_1,\dots,\varphi_N\in\fh$, and extension by linearity and continuity. By induction and 
(\ref{DefCre1})-(\ref{DefCre2})
\begin{align}
 \varphi_1\wedge\varphi_2\wedge\cdots\wedge\varphi_N=\e(\varphi_1)\e(\varphi_2)
	\cdots\e(\varphi_N)\Omega \label{AufbauSD}
\end{align}
for all $\varphi_1,\varphi_2,\dots,\varphi_N\in\fh$. If $\left\{\varphi_k\right\}_{k=1}^{\infty}\subseteq\fh$ 
is an orthonormal basis (ONB) of $\fh$, then for any $N\in\mathbbm{N}$
\begin{align}
 \left\{ \e(\varphi_{k_1})\cdots\e(\varphi_{k_N})\Omega |\ 1\leq k_1<k_2
      <\dots<k_N\right\}\subseteq\cF_{\mathrm{f}}^{(N)}[\fh] \label{ONBFN}
\end{align}
is an ONB of $\cF_{\mathrm{f}}^{(N)}[\fh]$, and
\begin{align}
 \left\{ \e(\varphi_{k_1})\cdots\e(\varphi_{k_N})\Omega 
      |\ N\in\mathbbm{N}_0,\ 1\leq k_1<k_2<\dots<k_N\right\}\subseteq\cF \label{ONBF}
\end{align}
is an ONB of $\cF$.
\newline
The adjoint operators $\v(f):=\left(\e(f)\right)^*\in\cB(\cF)$, with $f\in\fh$, are the annihilation operators. Note that, while
$f\mapsto\e(f)$ is linear, $f\mapsto\v(f)$ is antilinear. Together with the creation operators 
they fulfill the canonical anticommutation relations (CAR), i.e.,
\begin{align}
 \forall f,g\in\fh:\ \left\{\v(f),\e(g)\right\}=\left<f|g\right>_{\fh}\cdot\mathbbm{1}_{\cF},
	\quad \left\{\e(f),\e(g)\right\} = 0, \label{DefCAR}
\end{align}
where $\left\{A,B\right\}:=AB+BA$ denotes the anticommutator.\\
Moreover,
\begin{align}
 \forall f\in\fh:\ \v(f)\Omega=0, \label{DefAni2}
\end{align}
and $\left\{\e(f),\v(f)|\ f\in\fh\right\}$ is completely determined by (\ref{DefCre1}), (\ref{DefCAR}) and 
(\ref{DefAni2}), i.e., (\ref{DefCre2})-(\ref{ONBF}) follow from (\ref{DefCre1}), (\ref{DefCAR}) and (\ref{DefAni2}).
The creation and annihilation operators introduced here are a specific representation of the (abstract) 
CAR (\ref{DefCAR}), namely the Fock representation. For $\varphi_k$ being any element of a given ONB 
$\left\{\varphi_k\right\}_{k=1}^{\infty}\subseteq\fh$, we write 
\begin{align*}
  \e_k\equiv\e(\varphi_k),\ \v_k\equiv\v(\varphi_k).
\end{align*}
\newline
An important unbounded, self-adjoint and positive operator on $\cF$ is the number operator 
$\widehat{\mathbbm{N}}$ defined by 
\begin{align*}
 \widehat{\mathbbm{N}}\left(\e(f_1)\cdots\e(f_N)\Omega\right)
	:=N\cdot\e(f_1)\cdots\e(f_N)\Omega
\end{align*}
for any $f_1,\dots,f_N\in\fh$. It is not difficult to see that
\begin{align*}
 \widehat{\mathbbm{N}}=\sum\limits_{k=1}^{\infty}\e_k\v_k
\end{align*}
as a quadratic form, for any ONB $\left\{\varphi_k\right\}_{k=1}^{\infty}\subseteq\fh$.


\subsection{Density Matrices}
A positive trace class operator $\rho\in\cL_{+}^{1}(\cF)$ of unit trace, 
$\text{tr}_{\cF}\left\{\rho\right\}=1$, is called density matrix. Given a density matrix $\rho$, the map 
$A\mapsto \text{tr}_{\cF}\left\{\rho\, A\right\}$ defines a state, i.e., a normalized, linear, and positive 
functional on $\cB(\cF)\ni A$. If $\Psi\in\cF$ is a normalized vector, then 
$\left|\Psi\right>\left<\Psi\right|$ is a density matrix (of rank one) called pure state. In this paper we 
study fermion systems with a repulsive interaction and whose dynamics preserve the particle number. For this 
reason we restrict our attention to density matrices which commute with the particle
number operator and have a finite squared particle number expectation value,
\begin{align}
 \rho= \bigoplus_{N=0}^{\infty} \rho^{(N)},\quad\text{and}\quad 
  \big<\widehat{\mathbbm{N}}^2\big>_{\rho}<\infty \label{rhoF},
\end{align}
where here and henceforth we denote for any $A\in\cB(\fh)$ 
\begin{align*}
 \left<A\right>_{\rho}:=
    \TRF{\rho^{\frac{1}{2}}A\rho^{\frac{1}{2}}}.
\end{align*}
Note that, if $m\neq n,\ m,n\geq0$, then $\text{tr}_{\cF}\left\{\rho\,\e(f_1)
		\cdots\e(f_m)\v(g_1)\cdots\v(g_n)\right\}=0$ for any choice of 
		$f_1,\dots,f_m,g_1,\dots,g_n\in\fh$, 
due to (\ref{rhoF}).


\subsection{Reduced Density Matrices}
Given a density matrix $\rho\in\cL_{+}^{1}(\cF)$ subject to (\ref{rhoF}), we introduce two bounded 
operators, $\gamma_{\rho}\in\cB(\fh)$ and $\Gamma_{\rho}\in\cB(\fh\otimes\fh)$, by
\begin{align}
 \forall f,g\in\fh:\ \left<f |\,\gamma_\rho g\right>&:=\text{tr}_{\cF}\left\{\rho\, \e(g)
	      \v(f)\right\}\label{1pdm}
\end{align}
and
\begin{align}
 \forall f_1, f_2 ,g_1 , g_2 \in \fh:\ \left<f_1\otimes f_2 |\, \Gamma_\rho( g_1\otimes g_2)\right> 
:=\text{tr}_{\cF}\left\{\rho\,\e(g_2)\e(g_1)\v(f_1)\v(f_2)\right\}. \label{2pdm}
\end{align}
$\gamma_\rho$ is called the one-particle density matrix (1-pdm) and $\Gamma_\rho$ the two-particle density 
matrix (2-pdm) corresponding to $\rho$. For any ONB $\left\{\varphi_k\right\}_{k=1}^{\infty}\subseteq\fh$ we define the exchange operator 
$\mathrm{Ex}\in\mathcal{B}\left(\fh\otimes\fh\right)$ by 
\begin{align}
 \mathrm{Ex}:=\sum\limits_{k,l=1}^\infty\left|\vv_k\otimes\vv_l\right>\left<\vv_l\otimes\vv_k\right|,
\end{align}
such that $\mathrm{Ex}\left(f\otimes g\right)=g\otimes f$. Then the CAR leads to the antisymmetry property of $\Gamma_\rho$:
\begin{align}
 \mathrm{Ex}\,\Gamma_\rho=-\Gamma_\rho=\Gamma_\rho\,\mathrm{Ex}.
\end{align}
\ \\
The following properties of the 1-pdm and the 2-pdm are easily proven (we denote $\text{tr}_1 := \text{tr}_{\fh}$ 
and $\text{tr}_2 := \text{tr}_{\fh \otimes \fh}$):

\begin{lem}
 Let $\rho\in\cL_{+}^{1}(\cF)$ be a density matrix obeying (\ref{rhoF}). Then the following assertions hold true:
\begin{itemize}
\item[i)]$\gamma_\rho\in\cL_{+}^{1}(\fh),\quad 0\leq\gamma_\rho\leq \mathbbm{1},\quad \TRh{\gamma_\rho}
=\big<\widehat{\mathbbm{N}}\big>_{\rho},\quad\Gamma_\rho\in\cL_{+}^{1}(\fh\otimes\fh),\\ 0\leq
\Gamma_\rho\leq\big<\widehat{\mathbbm{N}}\big>_{\rho},\ \text{and}\quad \TRhh{\Gamma_\rho}
=\big<\widehat{\mathbbm{N}}\big(\widehat{\mathbbm{N}}-1\big)\big>_{\rho}$.
\item[ii)]
If $\mathrm{Ran}\left\{\rho\right\}\subseteq\cF_{\mathrm{f}}^{(N)}$, then, for all $f,g\in\fh$,
	\begin{align*}
	 \left<f|\,\gamma_\rho g\right>= \frac{1}{N-1}\sum\limits_{k=1}^{\infty}\left<f\otimes\varphi_k 
	  |\,\Gamma_\rho(g\otimes\varphi_k)\right>,
	\end{align*}
	where $\left\{\varphi_k\right\}_{k=1}^{\infty}\subseteq\fh$ is an ONB. 
\item[iii)] Furthermore,
	\begin{align*}
	\rho=\left|\e(\varphi_1)\cdots\e(\varphi_N)\Omega\right>
	\left<\e(\varphi_1)\cdots\e(\varphi_N)\Omega\right|\ \ \Leftrightarrow \ \ 
	\gamma_\rho=\sum\limits_{i=1}^{N}\left|\varphi_i\right>\left<\varphi_i\right|
	\end{align*}
	and, in this case,
	\begin{align*}
	 \Gamma_\rho=\left(1-\mathrm{Ex}\right)\left(\gamma_\rho\otimes\gamma_\rho\right).
	\end{align*}
\end{itemize}
\end{lem}


\subsection{Hamiltonian and Ground State Energy}
Recall from (\ref{Ham1}) that the Hamiltonian of an atom or molecule is given by 
\begin{align}
H^{(N)}(\underline{Z},\underline{R}):=\sum\limits_{n=1}^{N}\left(-\Delta_{x_{n}}-
\sum\limits_{k=1}^{K}\frac{Z_{k}}{\left|x_{n}-R_{k}\right|}\right)+\sum\limits_{1\leq n<m\leq N}
\frac{1}{\left|x_{n}-x_{m}\right|}. \label{Hamiltonian}
\end{align}
Choosing an ONB $\left\{\varphi_k\right\}_{k=1}^{\infty}\subseteq \fh=\text{L}^2\left(\mathbbm{R}^3\times
\mathbbm{Z}_2\right)$ such that $\left\{\varphi_k\right\}_{k=1}^{\infty}\subseteq \text{H}^2
\left(\mathbbm{R}^3\times\mathbbm{Z}_2\right)$, where $\text{H}^2\left(\mathbbm{R}^3\times\mathbbm{Z}_2\right)$ 
denotes the Sobolev space, we define 
\begin{align*}
 &h_{kl}:=\left<\varphi_{k}\Bigg|\left(-\Delta_x-\sum\limits_{k=1}^{K}\frac{Z_k}{\left|x-R_k\right|}\right)
\varphi_l \right>,\\
 &V_{kl;mn}:=\left< \varphi_{k} \otimes \varphi_{l} \bigg| \frac{1}{\left|x-y\right|}\left(\varphi_m \otimes 
\varphi_n \right) \right>,
\end{align*}
and
\begin{align*}
\widehat{\mathbbm{H}}:=\sum\limits_{k,l=1}^{\infty}h_{kl}\,\e_{k}\v_{l}+\sum\limits_{k,l,m,n=1}^{\infty}
V_{kl;mn}\,\e_{l}\e_{k}\v_{m}\v_{n}.
\end{align*}
Stability of matter ensures that $\widehat{\mathbbm{H}}+\mu \widehat{\mathbbm{N}}$ is a semibounded self-adjoint 
operator, provided $\mu<\infty$ is sufficiently large. Moreover, the Hamiltonian
of an atom or molecule can be viewed as
\begin{align*}
 H^{(N)}(\underline{Z},\underline{R})=
\widehat{\mathbbm{H}}\big|_{\cF_{\mathrm{f}}^{(N)}\left[\fh\right]},
\end{align*}
i.e., $H^{(N)}(\underline{Z},\underline{R})$ is the restriction of 
$\widehat{\mathbbm{H}}$ to $\cF_{\mathrm{f}}^{(N)}\left[\fh\right]$.
\newline
The ground state energy can now be reexpressed as
\begin{align*}
 \nonumber E_{\text{gs}}^{(N)}(\underline{Z},\underline{R})&=\inf\left\{  
\text{tr}_{\cF}\left\{\rho^{\frac{1}{2}}\widehat{\mathbbm{H}}\rho^{\frac{1}{2}}\right\} \Big|\  
\rho\in\cL_{+}^1(\cF),\ \widehat{\mathbbm{N}}\,\rho=N\rho,\ \TRF{\rho}=1\right\}\\
    &=\inf\left\{\mathcal{E}\left(\gamma_\rho,\Gamma_\rho\right) \Big|\ \rho\in\cL_{+}^1(\cF),\ 
\widehat{\mathbbm{N}}\,\rho=N\rho,\ \TRF{\rho}=1 \right\},
\end{align*}
where the energy functional is defined as
\begin{align*}
 \mathcal{E}\left(\gamma_\rho,\Gamma_\rho\right):=\text{tr}_1\left\{h\gamma_\rho\right\}+\frac{1}{2}\,
\text{tr}_2\left\{V\Gamma_\rho\right\}.
\end{align*}
We call $\left(\gamma,\Gamma\right)\in\cB(\fh)\times\cB(\fh\otimes\fh)$ 
{\underline{$N$-representable}} if there exists a density matrix $\rho\in\cL_+^1(\cF)$ 
with $\widehat{\mathbbm{N}}\,\rho=N\rho\ \text{and}\ \text{tr}_{\cF}\left\{\rho\right\}=1$
such that $\gamma=\gamma_\rho$ and $\Gamma=\Gamma_\rho$. Using the notion of $N$-representability, 
the ground state energy can be rewritten as
\begin{align*}
 E_{\text{gs}}^{(N)}=\inf\left\{\mathcal{E}\left(\gamma,\Gamma\right)\Big|\ 
\left(\gamma,\Gamma\right)\ \text{is}\ N\text{-representable}\right\}. 
\end{align*}
By Lemma 2.1, we have that
\begin{align*}
 E_{\text{hf}}^{(N)}=\inf\left\{\mathcal{E}\Big(\gamma,(1-\text{Ex})
(\gamma\otimes\gamma)\Big)\Big|\ \gamma=\gamma^*=\gamma^2,\ \text{tr}_1\left\{\gamma\right\}=N\right\},
\end{align*}
and Lieb's variational principle \cite{EHL, VBA} ensures that actually 
\begin{align*}
 E_{\text{hf}}^{(N)}=\inf\left\{\mathcal{E}\Big(\gamma,(1-\text{Ex})
(\gamma\otimes\gamma)\Big)\Big|\ 0\leq \gamma\leq 1,\ \text{tr}_1\left\{\gamma\right\}=N\right\}. 
\end{align*}


\section{G-, P- and Q-Conditions}

In this section we derive necessary conditions on $\left(\gamma,\Gamma\right)$ to be $N$-representable. 
To this end, we assume $N\in\mathbbm{N}$, $\gamma\in\cL^1(\fh)$ with $0\leq \gamma\leq 1$ 
and $\TRh{\gamma}=N$, $\Gamma\in\cL^1(\fh\otimes\fh)$, $\mathrm{Ex}\,\Gamma=\Gamma\,\mathrm{Ex}=-\,\Gamma$,
and we call $\left(\gamma, \Gamma\right)$ admissible in this case.
\begin{itemize}
\item[(P)] $\left(\gamma,\Gamma\right)$ fulfills the P-Condition
    \begin{align}
     :\Leftrightarrow \Gamma\geq 0.\label{PBed} 
     \end{align}
\item[(G)] $\left(\gamma,\Gamma\right)$ fulfills the G-Condition
    \begin{align}
   :\Leftrightarrow\forall A\in\cB(\fh):\ \TRhh{\left(A^*\otimes A\right)\left(\Gamma+\mathrm{Ex}
    \left(\gamma\otimes \mathbbm{1}\right)\right)}\geq\left|\TRh{A\gamma}\right|^2. \label{GBed}
    \end{align}
\item[(Q)] $\left(\gamma,\Gamma\right)$ fulfills the Q-Condition
    \begin{align}
 :\Leftrightarrow \Gamma+(1-\mathrm{Ex})(\mathbbm{1}\otimes\mathbbm{1}-\gamma\otimes\mathbbm{1}
    -\mathbbm{1}\otimes\gamma)\geq 0.\label{QBed}
    \end{align}
\end{itemize}

Our main result of this section is
\begin{thm}\label{mthm3}
 Let $\rho\in\cL^{1}(\cF)$ (not necessarily positive) such that $\TRF{\rho}=1$, 
$\TRF{\left|\rho\right|^{\frac{1}{2}}\widehat{\mathbbm{N}}^2\left|\rho\right|^{\frac{1}{2}}}<\infty$,
and that $\rho$ preserves the particle number, i.e., $\left[\widehat{\mathbbm{N}},\rho\right]=0$. Define $\gamma_\rho$ and $\Gamma_\rho$ by (\ref{1pdm}) and (\ref{2pdm}), respectively, and 
let $\left\{\vv_k\right\}_{k=1}^\infty\subseteq\fh$ be an ONB. 
Then the following two statements are equivalent.
\begin{itemize}
\item[(i)] If $\cP_r\in\cB(\cF)$ is a polynomial in $
    \left\{\e_k,\v_k\right\}_{k=1}^{\infty}$ of degree $r\leq 2$, then 
	\begin{align}
	\TRF{\rho\,\cP_r^*\cP_r}\geq 0.
	\end{align}
\item[(ii)] $\left(\gamma_\rho,\Gamma_\rho\right)$ is admissible and fulfills the G-, P- and 
	Q-Conditions.
\end{itemize}
\end{thm}

Before we turn to the proof of Theorem \ref{mthm3}, we establish its finite-dimensional analogue in Lemma \ref{lem3}
below. Theorem \ref{mthm3} then follows from Lemma \ref{lem3} by a limiting argument.

\begin{lem}\label{lem3}
 Let $\rho\in\cL^{1}(\cF)$ (not necessarily positive) such that $\TRF{\rho}=1$, 
$\TRF{\left|\rho\right|^{\frac{1}{2}}\widehat{\mathbbm{N}}^2\left|\rho\right|^{\frac{1}{2}}}<\infty$,
and that $\rho$ preserves the particle number, i.e., $\left[\widehat{\mathbbm{N}},\rho\right]=0$. 
Define $\gamma_\rho$ and $\Gamma_\rho$ by (\ref{1pdm}) and (\ref{2pdm}), respectively, and 
let $\left\{\vv_k\right\}_{k=1}^\infty\subseteq\fh$ be an ONB. Then the following statements are equivalent.
\begin{itemize}
\item[(i)] If $\cP_r\in\cB(\cF)$ is a polynomial in 
    $\left\{\e_k,\v_k\right\}_{k=1}^{\infty}$ of degree $r\leq 2$, then 
    \begin{align}
	\TRF{\rho\,\cP_r^*\cP_r}\geq 0 \label{posi}. 
    \end{align}
\item[(ii)] For any $\phi\in\mathrm{span}\left\{\vv_k|\ k\in\mathbbm{N}\right\}$,
      $\Psi\in\mathrm{span}\left\{\vv_k\otimes\vv_l|\ k,l\in\mathbbm{N}\right\}$, we have
      \begin{align}
	0\leq \left<\phi|\, \gamma_\rho\phi\right>\leq 1,\phantom{.............................} \label{preas}\\
	\left<\Psi |\, \Gamma_\rho \Psi\right>\geq 0,\phantom{..............................} \label{psdp}\\
	\left<\Psi |\left(\Gamma_\rho +(1-\mathrm{Ex})(\mathbbm{1}
	\otimes\mathbbm{1}-\gamma_\rho\otimes\mathbbm{1}-\mathbbm{1}\otimes\gamma_\rho)\right)\Psi\right>\geq 0,
	\label{psdq}
      \end{align}
    and, for all $A:=\sum\limits_{k,l=1}^{M}\alpha_{kl}\left|\vv_k\right>\left<\vv_l\right|$, $M<\infty$,
    $\left(\alpha_{kl}\right)_{k,l=1}^M\in\mathbbm{C}^{M\times M}$,
    \begin{align}
        \TRhh{\left(A^*\otimes A\right)\left(\Gamma_\rho+\mathrm{Ex}
    \left(\gamma_\rho\otimes \mathbbm{1}\right)\right)}\geq\left|\TRh{A\gamma_\rho}\right|^2. \label{psdg}
    \end{align}

\end{itemize}
\end{lem}

\begin{proof}
First we show {\textit{(i)}} $\Rightarrow$ {\textit{(ii)}}. 
The properties (\ref{preas})-(\ref{psdg}) of $\left(\gamma_\rho,\Gamma_\rho\right)$ can 
be checked by suitable choices of $\cP_r$.
\begin{itemize}
 \item[a)] The first inequality of (\ref{preas}) follows by choosing $\cP_1:=\sum\limits_{i=1}^M\overline{\alpha}_i \v_i$,
	where $\alpha_i\in\mathbbm{C}$ and $M<\infty$:
    \begin{align}
	\nonumber 0\leq \TRF{\rho\,\cP_1^*\cP_1}&=\sum\limits_{i,j=1}^M\alpha_i\overline{\alpha}_j\TRF{\rho\,\e_i \v_j}\\
	 &=\sum\limits_{j=1}^M\sum\limits_{i=1}^M\left<\alpha_j\varphi_j |\, \gamma_\rho(
		  \alpha_i\varphi_i)\right>=\left<\phi_M\big|\,\gamma_\rho\phi_M\right>, \label{g1}
    \end{align}
    with $\phi_M:=\sum\limits_{i=1}^{M}\alpha_i\varphi_i
    \in\mathrm{span}\left\{\vv_k|\ k\in\mathbbm{N}\right\}$. 
    The second inequality derives from the CAR and $\cP_1:=\sum\limits_{i=1}^{M}\alpha_i\e_i$:
    \begin{align}
	\nonumber 0\leq \TRF{\rho\,\cP_1^*\cP_1}&=\sum\limits_{i,j=1}^M\overline{\alpha}_i\alpha_j\TRF{\rho\,\v_i \e_j}\\
	\nonumber &=\sum\limits_{i,j=1}^M\overline{\alpha}_i\alpha_j\TRF{\rho\,(\delta_{ij}-\e_j\v_i)}\\
	\nonumber &=\sum\limits_{i=1}^M\sum\limits_{j=1}^M\left<\alpha_i\varphi_i |\,(\mathbbm{1}-\gamma_\rho) 
		  (\alpha_j\varphi_j)\right>\\
	&=\left<\phi_M\big|\,(\mathbbm{1}-\gamma_\rho)\phi_M\right>. \label{g2}
    \end{align}
\item[b)] Property (\ref{psdg}) is obtained by choosing $\cP_{\text{2}}:=\mu+\frac{1}{2}\sum\limits_{k,l=1}^M\alpha_{kl}
  \left(\e_k\v_l-\v_l\e_k\right)$ with $\mu,\alpha_{kl}\in\mathbbm{C},\ M<\infty$ and calculating $\TRF{\rho\,\cP_2^*\cP_2}$:
	\begin{align}
	\nonumber 0&\leq\text{tr}_{\cF}\bigg\{\rho\bigg(\mu+\frac{1}{2}\sum\limits_{k,l=1}^M\alpha_{kl}
		  (\e_k\v_l-\v_l\e_k)\bigg)^*\\
	\nonumber&\phantom{...............................................}\times\bigg(\mu+\frac{1}{2}\sum\limits_{m,n=1}^M\alpha_{mn}
		  (\e_m\v_n-\v_n\e_m)\bigg)\bigg\}\\
	\nonumber &=\text{tr}_{\cF}\bigg\{\rho\bigg(\frac{1}{2}\sum\limits_{k,l=1}^M\alpha_{kl}
		  (\e_k\v_l-\v_l\e_k)\bigg)^*\bigg(\frac{1}{2}\sum\limits_{m,n=1}^M\alpha_{mn}
		  (\e_m\v_n-\v_n\e_m)\bigg)\bigg\}\\
	\nonumber &\quad+2\,\mathfrak{Re}\bigg\{\overline{\mu}\,\text{tr}_{\cF}\bigg\{\frac{1}{2}\,\rho
		  \sum\limits_{k,l=1}^M\alpha_{kl}(\e_k\v_l-\v_l\e_k)\bigg\}\bigg\}+\left|\mu\right|^2.
	\end{align}
	Now we expand the brackets and use the CAR to reorder the annihilation and creation operators:
		\begin{align}
	\nonumber 0&\leq\sum\limits_{k,l,m,n=1}^M\overline{\alpha}_{kl}\alpha_{mn}\text{tr}_{\cF}
		  \bigg\{\rho\bigg(-\e_l\e_m\v_k\v_n+\delta_{km}\e_l\v_n-\frac{1}{2}\delta_{kl}\e_m\v_n\\
	\nonumber &\phantom{......................}-\frac{1}{2}\delta_{mn}\e_l\v_k+\frac{1}{4}\delta_{kl}
		  \delta_{mn}\bigg)\bigg\}\\
	\nonumber &\quad +2\,\mathfrak{Re}\bigg\{\overline{\mu}\sum\limits_{k,l=1}^M\alpha_{kl}\TRF{\rho\bigg(
		  \e_k\v_l-\frac{1}{2}\delta_{kl}\bigg)}\bigg\}+\left|\mu\right|^2.
	\end{align}
	Bearing
	$\alpha_{kl}=\left<\varphi_k |\,A\varphi_l\right>$ and $\TRF{\rho}=1$ in mind, we derive from the definitions of $\Gamma_\rho$ and $\gamma_\rho$ that
	\begin{align}
	\nonumber 0&\leq\sum\limits_{k,l,m,n=1}^M\bigg<\varphi_k\otimes\varphi_n\bigg|\bigg(-\Gamma_\rho+\mathbbm{1}
		  \otimes\gamma_\rho-\frac{1}{2}\mathrm{Ex}\left(\gamma_\rho\otimes\mathbbm{1}\right)
		  -\frac{1}{2}\mathrm{Ex}\left(\mathbbm{1}\otimes\gamma_\rho\right)\\
	\nonumber &\phantom{......................}+\frac{1}{4}\mathrm{Ex}\left(\mathbbm{1}\otimes\mathbbm{1}\right)
		  \bigg)(\varphi_m\otimes\varphi_l)\bigg>\\
	\nonumber &\phantom{..........................................................}\times\left<\varphi_l\otimes\varphi_m\big|
		  \left(A^*\otimes A\right)\left(\varphi_k\otimes\varphi_n\right)\right>\\
	\nonumber &\quad+2\,\mathfrak{Re}\left\{\overline{\mu}\,\sum\limits_{k,l=1}^M
		  \left<\varphi_k|\,A\varphi_l\right>\left<\varphi_l|\,\gamma_\rho\varphi_k\right>-\frac{\overline{\mu}}{2}
		  \sum\limits_{k=1}^M\left<\varphi_k|\,A\varphi_k\right>\right\}+\left|\mu\right|^2.
	\end{align}
	We can now perform the summations and arrive at
	\begin{align}
	\nonumber 0&\leq\mathrm{tr}_2\bigg\{\mathrm{Ex}\left(A^*\otimes A\right)\bigg(-\Gamma_\rho+\mathbbm{1}\otimes\gamma_\rho\\
	\nonumber&\phantom{......................}+\frac{1}{2}\mathrm{Ex}\bigg(\frac{1}{2}\mathbbm{1}\otimes\mathbbm{1}
		  -\mathbbm{1}\otimes\gamma_\rho-\gamma_\rho\otimes\mathbbm{1}\bigg)\bigg)\bigg\}\\
	\nonumber&\quad+2\,\mathfrak{Re}\left\{\overline{\mu}\,\TRh{A\gamma_\rho}-\frac{\overline{\mu}}{2}\,\TRh{A}\right\}
		  +\left|\mu\right|^2\\
	\nonumber&=\mathrm{tr}_2\bigg\{(A^*\otimes A)\bigg(\Gamma_\rho+(\mathbbm{1}\otimes\gamma_\rho)
		  \mathrm{Ex}+\frac{1}{4}\mathbbm{1}\otimes\mathbbm{1}-\frac{1}{2}(\mathbbm{1}\otimes\gamma_\rho
		  +\gamma_\rho\otimes\mathbbm{1})\bigg)\bigg\}\\
	&\quad+2\,\mathfrak{Re}\left\{\overline{\mu}\,\TRh{A\gamma_\rho}-\frac{\overline{\mu}}{2}\,\TRh{A}\right\}
		  +\left|\mu\right|^2. \label{Wahlmu}
	\end{align}
	Defining $s\in\mathbbm{C}$ by $\mu=:\left(s+\frac{1}{2}\right)\TRh{A}$, the inequality can be rewritten as
	\begin{align}
	\nonumber 0&\leq\text{tr}_{\cF}\bigg\{\rho\bigg(\mu+\frac{1}{2}\sum\limits_{k,l=1}^M\alpha_{kl}
		  (\e_k\v_l-\v_l\e_k)\bigg)^*\\
	\nonumber&\phantom{..............................}\times\bigg(\mu+\frac{1}{2}\sum\limits_{m,n=1}^M\alpha_{mn}
		  (\e_m\v_n-\v_n\e_m)\bigg)\bigg\}\\
	 &=\mathrm{tr}_{2}\Big\{(A^*\otimes A)\Big(\Gamma_\rho+\left|s
		  \right|^2\mathbbm{1}\otimes\mathbbm{1}
	+\overline{s}\,\mathbbm{1}\otimes\gamma_\rho+s\,
	\gamma_\rho\otimes\mathbbm{1}+(\mathbbm{1}\otimes\gamma_\rho)\mathrm{Ex}
		\Big)\Big\}. \label{notopt}
	\end{align}
	This inequality is valid for all $s$. We first assume that $\TRh{A}=\sum\limits_{i=1}^{M}\alpha_{ii}\neq 0$, 
	then  the choice $s:=-\,\frac{\TRh{A\gamma}}{\TRh{A}}$ optimizes the inequality. The conclusion is (\ref{psdg}):
	\begin{align}
	\TRhh{\left(A^*\otimes A\right)\left(\Gamma_\rho+\mathrm{Ex}\left(\gamma_\rho\otimes\mathbbm{1}\right)\right)}
	-\left|\TRh{A\gamma_\rho}\right|^2\geq 0. \label{GB}
	\end{align}
	Note that (\ref{notopt}) and (\ref{GB}) are equivalent because (\ref{GB}) implies (\ref{notopt}) by 
	\begin{align*}
	 \nonumber &\TRhh{\left(A^*\otimes A\right)\left(\Gamma_\rho+\mathrm{Ex}\left(\gamma_\rho\otimes\mathbbm{1}\right)\right)}
			  \geq\left|\TRh{A\gamma_\rho}\right|^2\\
	&\qquad\geq -\TRhh{\left(A^*\otimes A\right)\left(\left|s\right|^2\mathbbm{1}\otimes\mathbbm{1}
	+\overline{s}\,\mathbbm{1}\otimes\gamma_\rho+s\,\gamma_\rho\otimes\mathbbm{1}\right)}.
	\end{align*}
	Conversely, if $\TRh{A}=0$, the choice $\mu=
	-\,\TRh{A\gamma}$ in (\ref{Wahlmu}) leads directly to (\ref{GB}).
 \item[c)] Inserting $\cP_{\text{2}}:=\sum\limits_{k,l=1}^M\overline{\alpha}_{kl}\v_k\v_l$ into (\ref{posi}) yields inequality (\ref{psdp}):
	\begin{align*}
	\nonumber 0&\leq \text{tr}_{\cF}\bigg\{\rho\bigg(\sum\limits_{k,l=1}^M\overline{\alpha}_{kl}\v_k\v_l\bigg)^*
	\bigg(\sum\limits_{m,n=1}^M\overline{\alpha}_{mn}\v_m\v_n\bigg)\bigg\}\\
	\nonumber &=\sum\limits_{k,l,m,n=1}^M\alpha_{kl}\overline{\alpha}_{mn}\TRF{\rho\,\e_l\e_k\v_m\v_n}.
	\end{align*}
	By the definition of $\Gamma_\rho$ one finds
	\begin{align}
	\nonumber 0&\leq\sum\limits_{k,l,m,n=1}^M\alpha_{kl}\overline{\alpha}_{mn}\left<\varphi_m\otimes\varphi_n\big|\,
 		\Gamma_\rho( \varphi_k\otimes\varphi_l )\right>\\
	&=\left<\Psi_M\big|\,\Gamma_\rho\Psi_M\right>, \label{Pco}
	\end{align}
	where $\Psi_M:=\sum\limits_{i,j=1}^{M}\alpha_{ij}\left(\varphi_i\otimes\varphi_j\right)
	\in\mathrm{span}\left\{\vv_k\otimes\vv_l|\ k,l\in\mathbbm{N}\right\}$. 
\item[d)] Inequality (\ref{psdq}) follows from (\ref{posi}) by choosing $\cP_{\text{2}}:=\sum\limits_{k,l=1}^M\alpha_{kl}\e_k\e_l$:
	\begin{align*}
	\nonumber 0&\leq\text{tr}_{\cF}\bigg\{\rho\bigg(\sum\limits_{k,l=1}^M\alpha_{kl}\e_k\e_l\bigg)^*
		\bigg(\sum\limits_{m,n=1}^M\alpha_{mn}\e_m\e_n\bigg)\bigg\}\\
	\nonumber  &=\sum\limits_{k,l,m,n=1}^M\overline{\alpha}_{kl}\alpha_{mn}\TRF{\rho\,\v_l\v_k\e_m\e_n}.
	\end{align*}
	By normal-ordering using the CAR, one establishes the required relationship to $\Gamma_\rho$ and $\gamma_\rho$:
	\begin{align}
	\nonumber  0&\leq\sum\limits_{k,l,m,n=1}^M\overline{\alpha}_{kl}\alpha_{mn}\text{tr}_{\cF}
		\big\{\rho\,\big(\e_m\e_n\v_l\v_k-\delta_{ln}\e_m\v_k+\delta_{kn}\e_m\v_l+\delta_{lm}\e_n\v_k\\
	\nonumber  &\qquad\qquad\qquad\qquad\qquad\qquad\qquad
		 -\delta_{km}\e_n\v_l-\delta_{lm}\delta_{kn}+\delta_{km}\delta_{ln}\big)\big\}\\
	\nonumber  &=\Bigg<\sum\limits_{k,l=1}^M\alpha_{kl}(\varphi_l\otimes\varphi_k) \Bigg| 
		\Big(\Gamma_\rho+(1-\mathrm{Ex})(\mathbbm{1}\otimes\mathbbm{1}-
		\gamma_\rho\otimes\mathbbm{1}-\mathbbm{1}\otimes\gamma_\rho)\Big)\\
	\nonumber &\phantom{...........................................}
		\times\left(\sum\limits_{m,n=1}^M\alpha_{mn}(\varphi_n\otimes\varphi_m)\right)\Bigg>\\
	&=\Big<\Psi_M\Big|\Big(\Gamma_\rho+(1-\mathrm{Ex})(\mathbbm{1}\otimes\mathbbm{1}-
		\gamma_\rho\otimes\mathbbm{1}-\mathbbm{1}\otimes\gamma_\rho)\Big)\Psi_M\Big>. \label{Qco}
	\end{align}
\end{itemize}
Next we prove \textit{(ii)} $\Rightarrow$ \textit{(i)}. Thus, we assume (\ref{preas})-(\ref{psdg}).
\begin{itemize}
 \item[e)] A general polynomial of degree $r\leq 1$ is of the form 
$\cP_1=\sum\limits_{k=1}^{M}( \alpha_{k}\e_k+\beta_{k}\v_k)+\mu$ with $\mu,\alpha_k,\beta_k\in\mathbbm{C}$.
This means we have to consider
\begin{align} \TRF{\rho\,\cP_1^*\cP_1}&=\sum\limits_{k,l=1}^{M}\TRF{\rho\, (\alpha_{k}\e_k+\beta_{k}\v_k+\mu)^*
	 ( \alpha_{l}\e_l+\beta_{l}\v_l+\mu)}. \label{deg1}
\end{align}
We expand the product on the right side of (\ref{deg1}) and compute the traces, taking into account that $
\TRF{\rho\,\e_i}=\TRF{\rho\, \v_i}=\TRF{\rho\, \e_i\e_j}=\TRF{\rho\, \v_i\v_j}=0$ for every $i,j$
since $\rho$ preserves the particle number. Therefore, only three terms in (\ref{deg1}) are non-vanishing,
\begin{align*}
 \TRF{\rho\,\cP_1^*\cP_1}=\sum\limits_{k,l=1}^{M}\TRF{\rho\left((\alpha_k\e_k)^*\left(\alpha_l\e_l\right)
    +\left(\beta_k\v_k\right)^*\left(\beta_k\v_k\right)\right)}+\left|\mu\right|^2,
\end{align*}
where we additionally use $\TRF{\rho}=1$. The sum over the terms in braces is non-negative due to (\ref{g1}) and (\ref{g2}). The conclusion
is $\TRF{\rho\,\cP_1^*\cP_1}\geq 0$.
\item[f)]
 For $r\leq 2$ we have a general polynomial given by 
  \begin{align*}
   \nonumber\cP_2&=\nu+ \sum\limits_{k=1}^M\left(\alpha_k\e_k+\beta_k\v_k\right)+\sum\limits_{k,l=1}^M\alpha_{kl}\e_k\e_l+\sum\limits_{k,l=1}^M\beta_{kl}
      \v_k\v_l\\
   &\qquad\qquad+\sum\limits_{k,l=1}^M\kappa_{kl}\e_k\v_l+\sum\limits_{k,l=1}^M\eta_{kl}\v_k\e_l,
  \end{align*}
where $\nu,\alpha_k,\beta_k,\alpha_{kl},\beta_{kl},\kappa_{kl},\eta_{kl}\in\mathbbm{C}$, for all $1\leq k,l\leq M$.
Using the CAR, we rewrite $\cP_2$ as
\begin{align*}
 \cP_2=\cP_1+\cP_{2,\alpha}+\cP_{2,\beta}+\cP_{2,\eta},
\end{align*}
where
\begin{align*}
 &\cP_1:=\mu+\sum\limits_{k=1}^M\left(\alpha_k\e_k+\beta_k\v_k\right), &\cP_{2,\alpha}:=\sum\limits_{k,l=1}\alpha_{kl}\e_k\e_k,\\
 &\cP_{2,\beta}:=\sum\limits_{k=1}^M\beta_{kl}\v_k\v_l, &\cP_{2,\theta}:=\sum\limits_{k,l=1}^M\theta_{kl}\left(\e_k\v_l-\v_l\e_k\right),
\end{align*}
and 
\begin{align*}
 \mu:=\nu+\frac{1}{2}\sum\limits_{k=1}^M\left(\kappa_{kk}+\eta_{kk}\right),\quad \theta_{kl}:=\frac{1}{2}\left(\kappa_{kl}-\eta_{lk}\right).
\end{align*}
Then 
\begin{align*}
 \nonumber&\TRF{\rho\,\cP_2^*\cP_2}\\
  &\quad=\TRF{\rho\left(\cP_1^*+\cP_{2,\alpha}^*+\cP_{2,\beta}^*+\cP_{2,\eta}^*\right)\left(\cP_1+\cP_{2,\alpha}
	+\cP_{2,\beta}+\cP_{2,\theta}\right)}\\
  &\quad=\TRF{\rho\,\cP_1^*\cP_1}+\TRF{\rho\, \cP_{2,\alpha}^*\cP_{2,\alpha}}+\TRF{\rho\, \cP_{2,\beta}^*\cP_{2,\beta}}\\
    &\qquad\qquad\qquad	+\TRF{\rho\, \cP_{2,\theta}^*\cP_{2,\theta}},
\end{align*}
where we use that $\TRF{\rho\,\cP_a^*\cP_b}=0$ whenever $a\neq b$, since $\rho$ conserves the particle number. Now,
e) implies $\TRF{\rho\,\cP_1^*\cP_1}\geq 0$, (\ref{Qco}) yields $\TRF{\rho\,\cP_{2,\alpha}^*\cP_{2,\alpha}}\geq 0$
(see d)), (\ref{Pco}) yields $\TRF{\rho\,\cP_{2,\beta}^*\cP_{2,\beta}}\geq 0$ (see c)), and 
$\TRF{\rho\,\cP_{2,\theta}^*\cP_{2,\theta}}\geq 0$ follows from (\ref{GB}), see b). Hence, $\TRF{\rho\,\cP_2^*\cP_2}\geq 0$.
\end{itemize}
\end{proof}

Lemma \ref{lem3} is the algebraic part of the proof of Theorem \ref{mthm3}. To conclude, we have to extend the proof to infinite 
dimensions.\\

{\textit{Proof of Theorem \ref{mthm3}.}} Since \textit{(ii)} contains (\ref{preas})-(\ref{psdg}) of Lemma \ref{lem3},  the implication \textit{(ii)} 
$\Rightarrow$ \textit{(i)} is obvious. For \textit{(i)} $\Rightarrow$ \textit{(ii)}, let $\phi,\psi$ and $\Phi,\Psi$ be normalized vectors in $\fh$ and $\fh\otimes\fh$, respectively, 
and set $\alpha_i:=\left<\varphi_i|\phi\right>$, $\beta_i:=\left<\vv_i|\psi\right>$
and $\alpha_{ij}:=\left<\vv_i\otimes\vv_j|\Phi\right>$, $\beta_{ij}:=\left<\vv_i\otimes\vv_j|\Psi\right>$ for all $i,j\in\mathbbm{N}$. For $M\in\mathbbm{N}$, we
define the orthogonal projection $P_M:=\sum\limits_{k=1}^M\left|\vv_k\right>\left<\vv_k\right|$ and set
$\phi_M:=P_M\phi=\sum\limits_{i=1}^M\alpha_i\vv_i$ and $\Psi_M:=\left(P_M\otimes P_M\right)\Psi=
\sum\limits_{i,j=1}^M\beta_{ij}\left(\vv_i\otimes\vv_j\right)$.
The admissibility and the G-, P-, and Q-Condititions follow from Lemma \ref{lem3} as follows:
\begin{itemize}
\item[a)] For the 1-pdm we have $\|\gamma_\rho\|_\mathrm{op}<\infty$, since
	  \begin{align*}
	   \left|\left<\phi|\,\gamma_\rho\psi\right>\right|&=\left|\TRF{\rho\,\e(\psi)\v(\phi)}\right|\\
	    &\leq \TRF{\left|\rho\right|}\|\psi\|\|\phi\|=\TRF{\left|\rho\right|}<\infty
	  \end{align*}
	  by the Cauchy--Schwarz inequality and $c^*(\phi)c(\phi)\leq\mathbbm{1} \left<\phi|\phi\right>$. Afterwards, we infer by the 
	  triangle inequality
	\begin{align}
	 \nonumber\left|\left<\phi|\,\gamma_\rho\phi\right>-\left<\phi_M|\,\gamma_\rho\phi_M\right>\right|&= 
	    \left|\left<\phi-\phi_M|\,\gamma_\rho\phi\right>+\left<\phi_M|\,\gamma_\rho(\phi-\phi_M)\right>\right|\\
	 \nonumber&\leq   \left|\left<\phi-\phi_M|\,\gamma_\rho\phi\right>\right|+\left|\left<\phi_M|\,\gamma_\rho(\phi-\phi_M)\right>\right|\\
	   \nonumber &\leq \|\phi-\phi_M\|\left(\|\phi\|+\|\phi_M\|\right)\|\gamma_\rho\|_\mathrm{op}\\
	    &\leq2\,\|\phi-\phi_M\|\|\gamma_\rho\|_\mathrm{op}. \label{LimNorm}
	\end{align}
	  As $M\to\infty$, $\|\phi-\phi_M\|$ vanishes and we conclude with $\left<\phi_M|\,\gamma_\rho\phi_M\right>\geq 0$ that also
	  $\left<\phi|\,\gamma_\rho\phi\right>\geq 0$ for all $\phi\in\fh$. The same argument with $\gamma_\rho$ replaced by
	  $\mathbbm{1}-\gamma_\rho$ leads to 
	  $\left<\phi|\,(\mathbbm{1}-\gamma_\rho)\phi\right>\geq 0$ for all $\phi\in\fh$.
\item[b)] $\gamma_\rho\in\cL^1(\fh)$ follows by monotone convergence since $\gamma_\rho\geq 0$. For any ONB $\left\{\psi_i\right\}_{i=1}^{\infty}\subseteq\fh$
	and $\epsilon>0$, we have
	\begin{align*}
	&\sum\limits_{k=1}^{\infty} \left<\psi_k \big|\, \gamma_\rho\psi_k \right>\leq\sup\limits_{M}\bigg\{\sum\limits_{k=1}^{M} \TRF{\rho\, \e_k \v_k}\bigg\}\\
	&\qquad=
	    \sup\limits_{M}\bigg\{\mathrm{tr}_\cF\bigg\{\rho\sum\limits_{k=1}^{M}\e_k \v_k\bigg\}\bigg\}\\
	&\qquad=\left(N+\epsilon\right)\sup\limits_{M}\bigg\{\mathrm{tr}_\cF\bigg\{\rho\left(\frac{1}{\widehat{\mathbbm{N}}+\epsilon}\right)^\frac{1}{2}\bigg(\sum\limits_{k=1}^{M}\e_k \v_k\bigg)\left(\frac{1}{\widehat{\mathbbm{N}}+\epsilon}\right)^\frac{1}{2}\bigg\}\bigg\}\\
	&\qquad\leq \left(N+\epsilon\right)\TRF{\left|\rho\right|}<\infty,
	\end{align*}
	since $\widehat{\mathbbm{N}}\,\rho=N\rho$ and $\left\Vert\left(\frac{1}{\widehat{\mathbbm{N}}+\epsilon}\right)^\frac{1}{2}\bigg(\sum\limits_{k=1}^{M}\e_k \v_k\bigg)\left(\frac{1}{\widehat{\mathbbm{N}}+\epsilon}\right)^\frac{1}{2}\right\Vert_\mathrm{op}\leq 1$.
\item[c)] Thanks to b) we can compute $\sum\limits_{k=1}^{\infty} \left<\varphi_k \big|\, \gamma_\rho\varphi_k \right>$
	   using monotone convergence and $\widehat{\mathbbm{N}}\,\rho=N\rho $. This gives the trace of $\gamma_\rho$.
	 \begin{align*}
	    \TRh{\gamma_\rho}&:=\sum\limits_{k=1}^\infty\left<\vv_k|\,\gamma_\rho\vv_k\right>=\sum\limits_{k=1}^\infty\TRF{\rho\,\e_k\v_k}\\
		&=\mathrm{tr}_\cF\bigg\{\rho\sum\limits_{k=1}^\infty\e_k\v_k\bigg\}=\TRF{\rho\,\widehat{\mathbbm{N}}}=N.
	 \end{align*}
\item[d)] For any basis of $\fh$, the identities $\mathrm{Ex}\,\Gamma_\rho=\Gamma_\rho\,\mathrm{Ex}=-\,\Gamma_\rho$ are
	  a consequence of the definition of $\Gamma_\rho$ and the CAR.
\item[e)] We conclude from the definition of $\Gamma_\rho$ by the Cauchy--Schwarz inequality
	\begin{align*}
	   \left|\left<\Psi|\,\Gamma_\rho\Phi\right>\right|&=\bigg|\mathrm{tr}_\cF\bigg\{\rho\,\bigg(\sum\limits_{m,n=1}^\infty\beta_{mn}\v_n\v_m\bigg)^*
		\bigg(\sum\limits_{k,l=1}^\infty\alpha_{kl}\v_l\v_k\bigg)\bigg\}\bigg|\\
	    &\leq \TRF{\left|\rho\right|}\|\Psi\|\|\Phi\|=\TRF{\left|\rho\right|}<\infty.
	\end{align*}
	Therefore, we have $\|\Gamma_\rho\|_\mathrm{op}<\infty$. Afterwards, we infer analogously to (\ref{LimNorm})
	\begin{align*}
	 \left|\left<\Psi|\,\Gamma_\rho\Psi\right>-\left<\Psi_M|\,\Gamma_\rho\Psi_M\right>\right|\leq
		2\,\|\Psi-\Psi_M\|\|\Gamma_\rho\|_\mathrm{op},
	\end{align*}
	which tends to zero as $M\to\infty$ due to the definition of $\Psi_M$. With $\left<\Psi_M|\,\Gamma_\rho\Psi_M\right>\geq0$ from (\ref{Pco}),
	this means $\left<\Psi|\,\Gamma_\rho\Psi\right>\geq 0$ for all $\Psi\in\fh\otimes\fh$.
\item[f)] To prove that $\Gamma_\rho\in\cL^1(\fh\otimes\fh)$ we show that there is an ONB
	  $\left\{\psi_i\right\}_{i=1}^\infty\subseteq\fh$ such that $\sum\limits_{k,l=1}^{\infty}\left<\psi_k\otimes\psi_l \big|\,
		  \Gamma_\rho( \psi_k\otimes\psi_l)\right>$ is finite, again using $\Gamma_\rho\geq 0$. For any $\epsilon> 0$ we have, 
	  using $\widehat{\mathbbm{N}}\,\rho=N\rho$ and monotone convergence,
	 \begin{align*}
	  &\sum\limits_{k,l=1}^{\infty} \left< \psi_k \otimes \psi_l |\, \Gamma_\rho ( \psi_k \otimes \psi_l ) \right>
	  \leq \sup_M \bigg\{ \sum\limits_{k,l=1}^M \TRF{ \rho\, \e_l \v_k \e_k \v_l} \bigg\}\\
	  &\qquad= \sup_M \bigg\{ \mathrm{tr}_\cF\bigg\{ \rho\, \bigg[ \bigg( \sum\limits_{k=1}^M \e_k \v_k \bigg)^2 - \bigg( \sum\limits_{k=1}^M \e_k \v_k \bigg) \bigg] \bigg\} \bigg\}\\
	  &\qquad\leq \sup_M \bigg\{ \left(N^2 + \epsilon\right) \text{tr}_{\mathcal{F}} \bigg\{ \lvert \rho \rvert \left( \frac{1}{\widehat{\mathbbm{N}}^2 + \epsilon} \right)^{\frac{1}{2}} \bigg( \sum\limits_{k=1}^M c_k^* c_k \bigg)^2  \left( \frac{1}{\widehat{\mathbbm{N}}^2 + \epsilon} \right)^{\frac{1}{2}} \bigg\} \bigg\}\\
	  &\qquad\leq \left(N^2 + \epsilon\right) \TRF{ \lvert \rho \rvert} < \infty
	  \end{align*}
	  since $\sum\limits_{k=1}^M c_k^* c_k \geq 0$ and, due to $\left( \sum\limits_{k=1}^M c_k^* c_k \right)^2 \leq \widehat{\mathbbm{N}}^2$,
	  \begin{align*}
	  \left\Vert \left( \frac{1}{\widehat{\mathbbm{N}}^2 + \epsilon}  \right)^{\frac{1}{2}} \bigg(  \sum\limits_{k=1}^M c_k^* c_k  \bigg)^2  \left(  \frac{1}{\widehat{\mathbbm{N}}^2 + \epsilon}  \right)^{\frac{1}{2}} \right\Vert_\mathrm{op} \leq 1.
	  \end{align*}
\item[g)] To check (\ref{GB}) for any bounded $A$ (not necessarily of finite rank) we abbreviate $\Lambda_\mathrm{G}:=\Gamma_\rho+\mathrm{Ex}\left(\gamma_\rho\otimes\mathbbm{1}\right)-\gamma_\rho\otimes\gamma_\rho$
	   and set $A_M:=P_M A P_M$. Clearly, $A_M$ is 
	  of finite rank and we observe
	  \begin{align}
	  \nonumber  &\left|\TRhh{\left(A^*\otimes A\right)\Lambda_\mathrm{G}}-\TRhh{\left(A_M^*\otimes A_M\right)\Lambda_\mathrm{G}}\right|\\
	  \nonumber  &\qquad=\left|\TRhh{\left[\left(A-A_M\right)^*\otimes A + A_M^*\otimes\left(A-A_M\right)\right]\Lambda_\mathrm{G}}\right|\\
	  \nonumber  &\qquad=\Big|\mathrm{tr}_2\Big\{\Big[\left(P_M A^*P_M^\perp+P_M^\perp A^* P_M+P_M^\perp A^* P_M^\perp\right)\otimes A \\
	    &\qquad\qquad\qquad+ A_M^*\otimes\left(P_M^\perp A P_M + P_M A P_M^\perp+P_M^\perp A P_M^\perp\right)\Big]\Lambda_\mathrm{G}\Big\}\Big| \label{AusfG}
	  \end{align}
	  using $P_M^\perp:=\mathbbm{1}-P_M$. For $\left|\TRhh{\left(P_M^\perp A^* P_M\otimes A\right)\Lambda_\mathrm{G}}\right|$, 
	  for instance, we find
	  \begin{align}
	   \nonumber &\left|\TRhh{\left(P_M^\perp A^* P_M\otimes A\right)\Lambda_\mathrm{G}}\right|\\
	   \nonumber &\qquad=\big|\TRhh{\left(P_M^\perp A^* P_M\otimes A\right)\Gamma_\rho}+\TRh{P_M^\perp A^* P_M\gamma_\rho A}\\
	   \nonumber &\qquad\qquad\qquad   -\TRh{P_M^\perp A^* P_M\gamma_\rho}\TRh{ A\gamma_\rho}\big|\\
	   \nonumber &\qquad\leq\left|\TRhh{\left(P_M^\perp A^* P_M\otimes A\right)\Gamma_\rho}\right|+\left|\TRh{P_M^\perp A^* P_M\gamma_\rho A}\right|\\
		    &\qquad\qquad\qquad+\left|\TRh{P_M^\perp A^* P_M\gamma_\rho}\TRh{ A\gamma_\rho}\right|.
	  \end{align}
	  Since $\Gamma_\rho\geq 0$ due to the P-Condition (see e)), $P_M^\perp, P_M\geq0$ and $\widehat{\mathbbm{N}}\,\rho=N\rho$  we have on the one hand
	  \begin{align*}
	   \left|\TRhh{\left(P_M^\perp A^* P_M\otimes A\right)\Gamma_\rho}\right|&\leq \|A\|_\mathrm{op}^2\,\TRhh{\left(P_M^\perp\otimes 
			    \mathbbm{1}\right)\Gamma_\rho}\\
	      &=\left(N-1\right)\|A\|_\mathrm{op}^2\, \TRh{P_M^\perp\gamma_\rho},
	  \end{align*}
	  and on the other hand with $0\leq \gamma_\rho \leq \mathbbm{1}$ and $P_M\leq \mathbbm{1}$
	  \begin{align*}
	   \left|\TRh{P_M^\perp A^* P_M\gamma_\rho A}\right|
		    \leq \|A\|_\mathrm{op}^2\,\TRh{P_M^\perp\gamma_\rho},
	  \end{align*}
	  and
	  \begin{align*}
	  \left| \TRh{P_M^\perp A^* P_M\gamma_\rho}\TRh{ A\gamma_\rho}\right|\leq N\, \|A\|_\mathrm{op}^2\,\TRh{P_M^\perp\gamma_\rho}.
	  \end{align*}
	  Note that $\TRh{P^\perp_M\gamma_\rho}=\sum\limits_{k=M+1}^\infty\left<\vv_k|\,\gamma_\rho\vv_k\right>\to 0$,
	  as $M\to\infty$, since $\sum\limits_{k=1}^\infty\left<\vv_k|\,\gamma_\rho\vv_k\right>=\TRh{\gamma_\rho}=N$ is
	  convergent. Analogously, one finds that all terms on the right hand side of (\ref{AusfG}) tend to zero, as 
	  $M\to\infty$.
	  This, in turn, implies
	  $\TRhh{\left(A^*\otimes A\right)\Lambda_\mathrm{G}}\geq 0$ for any bounded $A$ since $\TRhh{\left(A_M^*\otimes A_M\right)\Lambda_\mathrm{G}}\geq 0$
	  due to (\ref{GB}).
\item[h)] Again by the Cauchy--Schwarz inequality we find
	\begin{align*}
	 &\Big|\Big<\Psi\Big|\Big(\Gamma_\rho +(1-\mathrm{Ex})(\mathbbm{1}
	\otimes\mathbbm{1}-\gamma_\rho\otimes\mathbbm{1}-\mathbbm{1}\otimes\gamma_\rho)\Big)\Phi\Big>\Big|\\
	&\qquad\qquad=\bigg|\mathrm{tr}_\cF\bigg\{\rho\,\bigg(\sum\limits_{m,n=1}^\infty\overline{\beta}_{mn}\v_m\v_n\bigg)
		\bigg(\sum\limits_{k,l=1}^\infty\overline{\alpha}_{kl}\v_k\v_l\bigg)^*\bigg\}\bigg|\\
	&\qquad\qquad\leq\TRF{\left|\rho\right|}\|\Psi\|\|\Phi\|
	\end{align*}
	and, therefore, $\|\Gamma_\rho +(1-\mathrm{Ex})(\mathbbm{1}
	\otimes\mathbbm{1}-\gamma_\rho\otimes\mathbbm{1}-\mathbbm{1}\otimes\gamma_\rho)\|_\mathrm{op}<\infty$.
	Following (\ref{LimNorm}) with $\Gamma_\rho +(1-\mathrm{Ex})(\mathbbm{1}
	\otimes\mathbbm{1}-\gamma_\rho\otimes\mathbbm{1}-\mathbbm{1}\otimes\gamma_\rho)$ instead of $\Gamma_\rho$,
	we arrive at
	\begin{align*}
	 \Big<\Psi\Big|\Big(\Gamma_\rho +(1-\mathrm{Ex})(\mathbbm{1}
	\otimes\mathbbm{1}-\gamma_\rho\otimes\mathbbm{1}-\mathbbm{1}\otimes\gamma_\rho)\Big)\Psi\Big>\geq 0\quad \forall\Psi\in\fh\otimes\fh.
	\end{align*}
\end{itemize}
$\left(\gamma_\rho,\Gamma_\rho\right)$ obeys the P-, G-, and Q-Conditions by e), g) and h). The admissibility is 
ensured in a) to d), and f). \qed
\phantom{}\\

A simple consequence of Theorem \ref{mthm3} is

\begin{cor}
 Let $N\in\mathbbm{N}$ and assume that $\left(\gamma,\Gamma\right)$ is $N$-representable. Then 
$\left(\gamma,\Gamma\right)$ is admissible and fulfills the G-, P- and Q-Conditions.
\end{cor}

\begin{proof}
 Since $\left(\gamma,\Gamma\right)$ is $N$-representable, there exists a density matrix $\rho\in\cL_+^1(\cF)$ 
with $\left(\gamma,\Gamma\right)\equiv\left(\gamma_\rho,\Gamma_\rho\right)$. By the last theorem, 
$\left(\gamma,\Gamma\right)$ then is admissible and fulfills the G-, P- and Q-Conditions. 
\end{proof}

\begin{rem}
 The G-Condition (\ref{psdg}) seems to be asymmetric in terms of $\gamma_\rho$. However, since 
$\TRhh{\left(A\otimes B\right)\Gamma_\rho}=\TRhh{\left(B\otimes A\right)\Gamma_\rho}$, it is easy to show that also 
$\TRhh{\left(A^*\otimes A\right)\left(\Gamma_\rho +\mathrm{Ex}
    \left(\mathbbm{1}\otimes \gamma_\rho\right)\right)}\geq\left|\TRh{A\gamma_\rho}\right|^2$ holds. Thus, we have a 
symmetrized, but weaker, G-Condition given by
\begin{align*}
 \TRhh{\left(A^*\otimes A\right)\left(\Gamma_\rho+\frac{1}{2}\mathrm{Ex}
    \left(\mathbbm{1}\otimes \gamma_\rho+\gamma_\rho\otimes\mathbbm{1}\right)\right)}\geq\left|\TRh{A\gamma_\rho}\right|^2.
\end{align*}
\end{rem}


\section{Correlation inequalities from G- and P-Conditions}
In \cite{VBA}, a lower bound on the difference of the ground state functional 
$\mathcal{E}\left(\gamma,\Gamma\right)$ and 
the Hartree--Fock functional $\mathcal{E}\left(\gamma,\left(1-\mathrm{Ex}\left(\gamma\otimes\gamma\right)\right)\right)$, 
i.e.,
\begin{align}
\TRhh{V\Gamma^{\left(\mathrm{T}\right)}}=\TRhh{V\left(\Gamma-\left(1-\mathrm{Ex}\right)
	      \left(\gamma\otimes\gamma\right)\right)}, \label{mSpur1}
\end{align} 
is derived using the decomposition of the potential $V$ according to Fefferman and de la Llave \cite{FDL}. It turns out that this 
decomposition is also useful to derive lower 
bounds only by means of $N$-representability. The main result of this section is the following theorem.

\begin{thm} \label{mthm4}
Let $X=X^*=X^2\in\cB(\fh)$ be an orthogonal projection on $\fh$. Assume that $\left(\gamma,\Gamma\right)$ is 
admissible and fulfills the G- and P-Conditions. Then
\begin{align}
\nonumber\TRhh{\left(X\otimes X\right)\Gamma^{\left(\mathrm{T}\right)}}&\geq-\,\TRh{X\gamma}\min\bigg\{1;\ 
	    38\,\TRh{X\left(\gamma-\gamma^2\right)}\\
&\qquad+4\Big[\TRh{X\left(\gamma-\gamma^2\right)}\Big(2+8\,\TRh{X\left(\gamma
		    -\gamma^2\right)}^2\Big)\Big]^\frac{1}{2}\bigg\}. \label{Hauptlem}
\end{align}
\end{thm}
\begin{proof}
The proof is carried out in several parts in the following sections. The first inequality is derived in Theorem \ref{mthm4a}. 
The second inequality follows from Theorems \ref{mthm4b}, \ref{mthm4c} and \ref{mthm4d}.
\end{proof}

In order to apply Theorem \ref{mthm4} to (\ref{mSpur1}), the potential $V$ on $\fh\otimes\fh$ is decomposed into an integral of a 
tensor product of two copies of the one-particle operator $X$. This decomposition is called Fefferman--de la Llave
identity.

\begin{lem}
$\text{For all}\ x\, ,y\, \in\mathbbm{R}^3,\ x\neq y$, one has
\begin{align}
\frac{1}{\left|x-y\right|}=\int\limits_{0}^{\infty}\frac{\mathrm{d}r}{\pi r^5}\int\limits_{\mathbbm{R}^3}\mathrm{d}^3 z\ 
 	      \chi_{B\left(z,r\right)}\left(x\right)\, \chi_{B\left(z,r\right)}\left(y\right), \label{FDLD}
\end{align}
where $\chi_{B\left(z,r\right)}$ is the characteristic function of the ball $B\left(z,r\right):=\left\{x\in\mathbbm{R}^3
	\,|\ \left|x-z\right|\leq r\right\}$ of radius $r>0$ centered at $z\in\mathbbm{R}^3$.
\end{lem}

The proof of the decomposition can be found in the original work of Fefferman and de la Llave in \cite{FDL}. In \cite{HS1}, Hainzl and Seiringer have
derived sufficient conditions on a pair potential $V:\mathbbm{R}^n\to\mathbbm{R}$ so as to admit a decomposition
of the form (\ref{FDLD}).

\begin{rem}
The multiplication operator corresponding to $\chi_{B\left(z,r\right)}$ is denoted by 
$X_{r,z}\equiv X$. Clearly, 
\begin{align} 
\cB(\fh)\ni X=X^*=X^2 \label{defX}
 \end{align}
is an orthogonal projection.
\end{rem}

Instead of (\ref{mSpur1}) we consider from now on
\begin{align}
\TRhh{\left(X\otimes X\right)\Gamma^{\left(\mathrm{T}\right)}}
=\TRhh{\left(X\otimes X\right)\left(\Gamma-\left(1-\mathrm{Ex}\right)
	      \left(\gamma\otimes\gamma\right)\right)}. \label{Spur1}
\end{align}

A first estimation of this quantity is immediately
obtained by applying the G-Condition directly on $\TRhh{\left(X\otimes X\right)\Gamma}$. This yields the first inequality 
of (\ref{Hauptlem}).

\begin{thm}\label{mthm4a}
Let $X$ be as in (\ref{defX}). 
Assume that $\left(\gamma,\Gamma\right)$ is admissible and fulfills the G-Condition. Then
\begin{align*}
\qquad\TRhh{\left(X\otimes X\right)\Gamma^{\left(\mathrm{T}\right)}}\geq -\,\TRh{X\gamma}.
\end{align*}
\end{thm}
\begin{proof}
As mentioned, we apply the G-Condition (\ref{GBed}) with $A^*=A:=X$ directly on $\TRhh{\left(X\otimes X\right)\Gamma}$. 
The HF-part is carried out:
\begin{align}
\nonumber &\TRhh{\left(X\otimes X\right)\left(\Gamma-\left(1-\mathrm{Ex}\right)\left(\gamma\otimes\gamma\right)\right)}\\
\nonumber &\qquad\geq \left(\TRh{X\gamma}\right)^2-\TRh{X\gamma}-\left(\TRh{X\gamma}\right)^2
	+\TRh{X\gamma X\gamma}\\
&\qquad \geq -\,\TRh{X\gamma}. \label{abs1}
\end{align} 
The last inequality follows from $\TRh{X\gamma X\gamma}=\TRh{X\gamma X\gamma X}\geq 0$.
\end{proof}

The goal of the next sections is an estimation of (\ref{Spur1}) in 
terms of $\TRh{X\left(\gamma-\gamma^2\right)}$.


\subsection{Preparation}
A crucial step in \cite{VBA} is the decomposition of the spectrum of $\gamma$ into eigenvalues which are larger
than $\frac{1}{2}$, and those which are smaller or equal $\frac{1}{2}$. Following this step, the decomposition is denoted  by two 
orthogonal projections, $P$ and $P^{\perp}$ (a comparable strategy was also used by Graf and Solovej in \cite{SVG}). The first one, $P$, projects on the space which is spanned by the eigenvectors 
of $\gamma$ corresponding to eigenvalues larger than $\frac{1}{2}$. The se\-cond one treats the eigenvectors with
eigenvalues smaller or equal $\frac{1}{2}$. Furthermore, the eigenvectors of $\gamma$, 
$\left\{\varphi_i\,|\ \gamma\varphi_i=\lambda_i\varphi_i\right\}_{i=1}^{\infty}$, are used as an ONB
of $\fh$ which we mainly refer to. In this basis the two projections can be defined straightforwardly.

\begin{defn}
On $\fh$, the orthogonal projections $P$ and $P^{\perp}$ are defined by
\begin{align}
P:=\mathbbm{1}\left[\gamma>\frac{1}{2}\right]=\sum\limits_{k>\frac{1}{2}}\left|\varphi_k\right>\left<\varphi_k\right|
    \quad\text{and}\quad P^{\perp}:=\mathbbm{1}\left[\gamma\leq\frac{1}{2}\right]=\sum\limits_{k\leq\frac{1}{2}}
    \left|\varphi_k\right>\left<\varphi_k\right|. \label{defP}
\end{align}
\end{defn}

Here, the summation over ''$k>\frac{1}{2}$" denotes the summation over $\left\{k\,\big|\ \lambda_k>\frac{1}{2}\right\}$ 
and for  ''$k\leq\frac{1}{2}$" analogously. Obviously,
\begin{align*}
 P+P^{\perp}=\mathbbm{1},\quad P P^{\perp}=P^{\perp}P=0,\quad\ P\gamma=\gamma P\quad \text{and}\quad P^{\perp}\gamma=\gamma P^{\perp}.
\end{align*}
Moreover, the projections are bounded from above.

\begin{lem}
For $P$ and $P^{\perp}$ defined in (\ref{defP})
\begin{align}
P\leq 2\gamma\quad\text{and}\quad P^{\perp}\leq 2\left(\mathbbm{1}-\gamma\right) \label{Proj2}
  \end{align}
hold true.
\end{lem}

Note that, since $\mathrm{rk}\left\{P\right\}\leq 2N$, $P$ is of finite rank and, hence, trace class.

\begin{proof}
Using the definition of the projections together with $0\leq\gamma\leq \mathbbm{1}$, one finds for $P$:
\begin{align*}
P=\sum\limits_{k>\frac{1}{2}}\left|\varphi_k\right>\left<\varphi_k\right|\leq\sum\limits_{k>\frac{1}{2}}2 
      \lambda_k \left|\varphi_k\right>\left<\varphi_k\right|\leq2\sum\limits_{k=1}^\infty\lambda_k\left|\varphi_k\right>
      \left<\varphi_k\right|=2\gamma,
\end{align*}
and for $P^{\perp}$:
\begin{align*}
 P^{\perp}=\sum\limits_{k\leq\frac{1}{2}}\left|\varphi_k\right>\left<\varphi_k\right| 
		\leq \sum\limits_{k\leq\frac{1}{2}}2\left(1-\lambda_k\right)\left|\varphi_k\right>\left<\varphi_k\right| 
      \leq 2\left(\mathbbm{1}-\gamma\right).
\end{align*}
\end{proof}

Thanks to $P^\perp + P = \mathbbm{1}$ we can expand $\TRhh{\left(X\otimes X\right)\Gamma}$ into three parts to have expressions
on which we can apply the conditions on $\left(\gamma,\Gamma\right)$. We denote this three parts by Main Part (MP), 
Remainder (R) and Main Error Term (MET).

\begin{lem}
Let $X$, and $P$ and $P^{\perp}$ be as defined in (\ref{defX}) and (\ref{defP}), respectively. Then
\begin{align}
\nonumber &\TRhh{\left(X\otimes X\right)\Gamma}\\
\nonumber &\quad	 = \TRhh{\left(PXP\otimes PXP\right)\Gamma} 
				+ 4\,\mathfrak{Re}\left\{\TRhh{\left(PXP\otimes P^{\perp}XP\right)\Gamma}\right\}\\
\nonumber &\qquad\qquad\qquad\qquad\qquad\qquad\qquad\qquad\qquad
			 + 2\,\TRhh{\left(PXP^{\perp}\otimes P^{\perp}XP\right)\Gamma}\\
\nonumber &\quad\quad + \TRhh{\left(P^{\perp}XP^{\perp}\otimes P^{\perp}XP^{\perp}\right)\Gamma} 
				+ 4\,\mathfrak{Re}\left\{\TRhh{\left(P^{\perp}XP\otimes P^{\perp}XP^{\perp}\right)\Gamma}
							\right\}\\
\nonumber &\qquad\qquad\qquad\qquad\qquad\qquad\qquad\qquad\qquad
			+ 2\,\TRhh{\left(P^{\perp}XP^{\perp}\otimes PXP\right)\Gamma}\\
		  &\qquad + 2\,\mathfrak{Re}\left\{\TRhh{\left(PXP^{\perp}\otimes PXP^{\perp}\right)\Gamma}
		  	\right\}. \label{MP+R}
\end{align}
\end{lem}
\begin{proof}
After replacing the identity operator on each side of $X$ in each factor of the tensor product $X\otimes X$ by 
$\left(P+P^{\perp}\right)$, one can expand the r.h.s. of
\begin{align*} 
\nonumber&\TRhh{\left(X\otimes X\right)\Gamma}\\
\nonumber&\qquad=\TRhh{\left(\left(P+P^{\perp}\right)X\left(P+P^{\perp}\right)\otimes 
    \left(P+P^{\perp}\right)X\left(P+P^{\perp}\right)\right)\Gamma}.
\end{align*} 
Using $\TRhh{\left(A\otimes B\right)\Gamma}=\TRhh{\left(B\otimes A\right)\Gamma}$, which is
a consequence of $\mathrm{Ex}\,\Gamma\,\mathrm{Ex}=\Gamma$, one arrives at the assertion after rearranging.
\end{proof}

Afterwards, we collect the terms of (\ref{MP+R}) in a suitable way. Note that compared to \cite{VBA} the definitions of 
the Main Part and the Remainder are slightly changed.

\begin{defn}
The term
\begin{align*}
\nonumber T_{\mathrm{MP}}&:=\TRhh{\left(PXP\otimes PXP\right)\Gamma} 
				+ 4\,\mathfrak{Re}\left\{\TRhh{\left(PXP\otimes P^{\perp}XP\right)\Gamma}\right\}\\
			&\qquad+4\,\TRhh{\left(PXP^{\perp}\otimes P^{\perp}XP\right)\Gamma}
\end{align*}
is called Main Part,
\begin{align*}
\nonumber T_{\mathrm{R}}&:=\TRhh{\left(P^{\perp}XP^{\perp}\otimes P^{\perp}XP^{\perp}\right)\Gamma} 
	+ 4\,\mathfrak{Re}\left\{\TRhh{\left(P^{\perp}XP\otimes P^{\perp}XP^{\perp}\right)\Gamma}\right\}\\
    &\qquad+2\,\TRhh{\left(P^{\perp}XP^{\perp}\otimes PXP\right)\Gamma}-2\,\TRhh{\left(PXP^{\perp}\otimes 
	      P^{\perp}XP\right)\Gamma}
\end{align*}
is called Remainder, and
\begin{align*}
T_{\mathrm{MET}}:= 2\,\mathfrak{Re}\left\{\TRhh{\left(PXP^{\perp}\otimes PXP^{\perp}\right)\Gamma}\right\}
\end{align*}
is called Main Error Term.
\end{defn} 

One estimate is used more than once when considering the terms in the Remainder and Main Error Term. This
estimate requires the following lemma.

\begin{lem} \label{lem4a}
Let $\left\{\psi_i\right\}_{i=1}^{\infty}\subseteq\fh$ be an ONB, $Q=Q^*=Q^2$, $Q^{\perp}:=\mathbbm{1}-Q$, and $Y=Y^*=Y^2\in\cB(\fh)$ orthogonal projections.
For $r,s\in\mathbbm{N}$ define
\begin{align}
B\left(r,s\right):=\left| QY \psi_r\right>\left< Q^{\perp} Y\psi_s\right|\in\cB(\fh).
\end{align}
Then one has
\begin{align}
\nonumber &\sum\limits_{r,s=1}^{\infty}\TRhh{\left(B^*\left(r,s\right)\otimes
	B\left(r,s\right)\right)
	\left(\Gamma+\mathrm{Ex}\left(\gamma\otimes\mathbbm{1}\right)\right)}\\
		&\qquad\qquad=\TRhh{\left(QYQ\otimes Q^{\perp}YQ^{\perp}\right)
		\left(-\Gamma+\mathbbm{1}\otimes\gamma\right)}. \label{HIL}
\end{align}
\end{lem}
\begin{proof}
Denoting $K:=\sum\limits_{r,s=1}^{\infty}\TRhh{\left(B^*\left(r,s\right)\otimes
	B\left(r,s\right)\right)
	  \left(\Gamma+\mathrm{Ex}\left(\gamma\otimes\mathbbm{1}\right)\right)}$, we calculate the trace using the 
	ONB $\left\{\psi_i\right\}_{i=1}^{\infty}\subseteq\fh$.
\begin{align*}
\nonumber K&=\sum\limits_{r,s=1}^{\infty}\sum\limits_{k,l,m,n=1}^{\infty}\left<\psi_k\otimes\psi_l\big|
 \left(B^*\left(r,s\right)\otimes B\left(r,s\right)\right)\left(\psi_m\otimes\psi_n\right)\right>\\
\nonumber &\qquad\qquad\qquad\qquad\times\left<\psi_m\otimes\psi_n\big|\left(\Gamma +\mathrm{Ex}
	\left(\gamma\otimes\mathbbm{1}\right)\right)\left(\psi_k\otimes\psi_l\right)\right>.
\end{align*}
In the next step, the definition of $B\left(r,s\right)$ together with 
$\left(A\right)_{ij}:=\left<\psi_i|\,A\psi_j\right>$, for any $A\in\cB(\fh)$,
can be used to write
\begin{align*}
\nonumber K&=\sum\limits_{r,s=1}^{\infty}\sum\limits_{k,l,m,n=1}^{\infty}\left(YQ\right)_{rm}
	\left(Q^{\perp}Y\right)_{ks}\left(QY\right)_{lr}
	\left(YQ^{\perp}\right)_{sn}\\
\nonumber &\qquad\qquad\qquad\qquad\times\left<\psi_m\otimes\psi_n\big|\left(\Gamma +\mathrm{Ex}
	\left(\gamma\otimes\mathbbm{1}\right)\right)\left(\psi_k\otimes\psi_l\right)\right>.
\end{align*}
Performing the summation over $r$ and $s$ leads to
\begin{align*}
\nonumber K&=\sum\limits_{k,l,m,n=1}^{\infty}\left(QYQ\right)_{lm}
\left(Q^{\perp}YQ^{\perp}\right)_{kn}\\
\nonumber &\qquad\qquad\qquad\qquad\times\left<\psi_m\otimes\psi_n\big|\left(\Gamma +\mathrm{Ex}
	\left(\gamma\otimes\mathbbm{1}\right)\right)\left(\psi_k\otimes\psi_l\right)\right>\\
\nonumber &=\sum\limits_{k,l,m,n=1}^{\infty}\left<\psi_l\otimes\psi_k\big|\left(QYQ\otimes 
	Q^{\perp}YQ^{\perp}\right)\left(\psi_m\otimes\psi_n\right)\right>\\
\nonumber &\qquad\qquad\qquad\qquad\times\left<\psi_m\otimes\psi_n\big|\left(\Gamma +\mathrm{Ex}
	\left(\gamma\otimes\mathbbm{1}\right)\right)\left(\psi_k\otimes\psi_l\right)\right>.
\end{align*}
The summation over $m$ and $n$ can also be carried out. Finally, the summation over $k$ and $l$ gives
\begin{align*}
\nonumber K&=\sum\limits_{k,l=1}^{\infty}
	\left<\psi_k\otimes\psi_l\big|\left(\mathrm{Ex}\left(QYQ\otimes Q^{\perp}YQ^{\perp}\right)
	\left(\Gamma+\mathrm{Ex}\left(\gamma\otimes\mathbbm{1}\right)\right)\right)\left(\psi_k\otimes\psi_l\right)\right>\\
&=\TRhh{\left(QYQ\otimes Q^{\perp}YQ^{\perp}\right)\left(-\,\Gamma+\mathbbm{1}\otimes\gamma\right)},
\end{align*}
using the cyclicity of the trace, $\Gamma\,\mathrm{Ex}=-\,\Gamma$, and $\mathrm{Ex}\left(\mathbbm{1}\otimes\gamma\right)\mathrm{Ex}=\gamma\otimes\mathbbm{1}$.
\end{proof}
\begin{rem}
By changing the definition of $B\left(r,s\right)$ it is also possible to treat, 
for example, $QYQ\otimes QYQ$ similarly. However, it is important to notice that 
$\sum\limits_{r,s}B^*\left(r,s\right)\otimes B\left(r,s\right)$ is, in general, indefinite.
In fact, the trace of $B\left(r,s\right)$ is vanishing in our case and so is the trace of $B^*\left(r,s\right)\otimes
B\left(r,s\right)$. Hence, $B^*\left(r,s\right)\otimes
B\left(r,s\right)$ is either indefinite or zero. Furthermore, since $P^\perp\ \text{and}\ P$ commute with $\gamma$, we also 
have $\TRh{B\left(r,s\right)\gamma}=0$.
\end{rem}

A consequence of Lemma \ref{lem4a} is a key inequality for proving the estimate on the Remainder and
the Main Error Term. This inequality is given in (\ref{Fundsatz}).

\begin{lem}
Let $X$, and $P$ and $P^{\perp}$ be as defined in (\ref{defX}) and (\ref{defP}), respectively.
Assume that $\left(\gamma,\Gamma\right)$ is admissible and fulfills the G-Condition. Then
\begin{align}
\quad\TRhh{\left(PXP\otimes P^{\perp}XP^{\perp}\right)\Gamma}\leq 4\,\TRh{X\gamma}
  \TRh{X\left(\gamma-\gamma^2\right)}. \label{Fundsatz}
\end{align}
\end{lem}
\begin{proof}
First, we observe that (\ref{HIL}) with $Y=X$ and $Q=P$ and the G-Condition immediately lead to
	\begin{align*}
	\nonumber 0&\leq\sum\limits_{r,s=1}^{\infty}TRhh{\left( B^*\left(r,s\right)\otimes B\left(r,s\right)\right)
		\left(\Gamma+\mathrm{Ex}\left(\gamma\otimes
		\mathbbm{1}\right)\right)}\\
	  &=\TRhh{\left(PXP\otimes P^{\perp}XP^{\perp}\right)
	 	\left(-\,\Gamma+\left(\mathbbm{1}\otimes\gamma\right)\right)}.
	\end{align*}
	Consequently,
	\begin{align*}
	\nonumber\TRhh{\left(PXP\otimes P^{\perp}XP^{\perp}\right)\Gamma}&\leq \TRhh{\left(PXP\otimes 
	      P^{\perp}XP^{\perp}\right)\left(\mathbbm{1}\otimes\gamma\right)}\\
		&=\TRh{PX}\TRh{P^{\perp}X\gamma}.
	\end{align*}
Secondly, we permute the arguments in the trace cyclically and use that $\gamma$ is trace class and 
$PX$ and $P^{\perp}X$ are bounded.
Then we use  (\ref{Proj2}) to estimate the projections:
\begin{align*}
\TRh{PX}\TRh{P^{\perp}X\gamma}&=\TRh{XPX}\TRh{X\sqrt{\gamma}P^{\perp}\sqrt{\gamma}X}\\
      &\leq 4\,\TRh{X\gamma}\TRh{X\left(\gamma-\gamma^2\right)}
\end{align*}
since $\gamma P^\perp=\sqrt{\gamma}P^\perp\sqrt{\gamma}$.
\end{proof}

\begin{rem}
From $\TRhh{\left(A\otimes B\right)\Gamma}=\TRhh{\left(B\otimes A\right)\Gamma}$, for $A,B\in\cB(\fh)$, one 
directly can conclude
\begin{align}
\TRhh{\left(P^{\perp}XP^{\perp}\otimes PXP\right)\Gamma}\leq 4\,\TRh{X\gamma}\TRh{X\left(\gamma-\gamma^2\right)}. 
												    \label{Fundsatz2}
\end{align}
\end{rem}


\subsection{Estimation of the Remainder}

Now we consider the Remainder of (\ref{MP+R}):
\begin{align*}
\nonumber T_{\mathrm{R}}&:=\TRhh{\left(P^{\perp}XP^{\perp}\otimes P^{\perp}XP^{\perp}\right)\Gamma} 
 				+ 4\,\mathfrak{Re}\left\{\TRhh{\left(P^{\perp}XP\otimes P^{\perp}XP^{\perp}\right)
 							\Gamma}\right\}\\
			&\qquad+2\,\TRhh{\left(P^{\perp}XP^{\perp}\otimes PXP\right)\Gamma}
				-2\,\TRhh{\left(PXP^{\perp}\otimes P^{\perp}XP\right)\Gamma}.
\end{align*}
The first three terms, summed up in $T_{\mathrm{R}_1}$ and the last term, called $T_{\mathrm{R}_2}$, are treated 
separately to derive a lower bound.

\begin{lem} \label{lem4b}
Let $X$, and $P$ and $P^{\perp}$ be as defined in (\ref{defX}) and (\ref{defP}), respectively.
Assume that $\left(\gamma,\Gamma\right)$ is admissible and fulfills the G- and P-Conditions. Then
\begin{align*}
\nonumber T_{\mathrm{R}_1}&:=\TRhh{\left(P^{\perp}XP^{\perp}\otimes P^{\perp}XP^{\perp}\right)\Gamma} 
			      +2\,\TRhh{\left(P^{\perp}XP^{\perp}\otimes PXP\right)\Gamma}\\
\nonumber&\quad\qquad\qquad\qquad\qquad+ 4\,\mathfrak{Re}\left\{\TRhh{\left(P^{\perp}XP^{\perp}\otimes PXP^{\perp}\right)\Gamma}\right\}\\
&\geq -\,8\,\TRh{X\gamma}\TRh{X\left(\gamma-\gamma^2\right)}.
\end{align*}
\end{lem}
\begin{proof}
First, we use $\mathfrak{Re}\left(\zeta\right)\geq -\left|\zeta\right|$, for any complex number $\zeta$, and
$P^{\perp}XP\otimes P^{\perp}XP^{\perp}=\left(P^{\perp}X\otimes P^{\perp}X\right)\left(XP\otimes X P^{\perp}\right)$ to infer
\begin{align*}
\nonumber  T_{\mathrm{R}_1}&=\TRhh{\left(P^{\perp}XP^{\perp}\otimes P^{\perp}XP^{\perp}\right)\Gamma}
	  +2\,\TRhh{\left(PXP\otimes P^{\perp}XP^{\perp}\right)\Gamma}\\
\nonumber &\qquad\qquad\qquad+4\,\mathfrak{Re}\left\{\TRhh{\left(P^{\perp}XP\otimes P^{\perp}XP^{\perp}\right)\Gamma}\right\}\\
\nonumber & \geq\TRhh{\left(P^{\perp}XP^{\perp}\otimes P^{\perp}XP^{\perp}\right)\Gamma}
	+2\,\TRhh{\left(PXP\otimes P^{\perp}XP^{\perp}\right)\Gamma}\\
\nonumber&\qquad\qquad\qquad-4\left|\TRhh{\left(P^{\perp}X\otimes P^{\perp}X\right)\left(XP\otimes X P^{\perp}\right)\Gamma}\right|.
\end{align*}
Then we use that $\left(A,B\right):=
\TRhh{A^* B\,\Gamma}$ defines a positive semidefinite Hermitian form on $\cB(\fh\otimes\fh)$, due to $\Gamma\geq0$, which is the P-Condition. Hence, the Cauchy--Schwarz inequality
\begin{align}
\left|\left(A,B\right)\right|\leq\left(A,A\right)^{\frac{1}{2}}
	\left(B,B\right)^{\frac{1}{2}} \label{CSIN}
\end{align} 
holds, and we arrive at
\begin{align*}
\nonumber& T_{\mathrm{R}_1}\geq\TRhh{\left(P^{\perp}XP^{\perp}\otimes P^{\perp}XP^{\perp} \right)\Gamma}
		+2\,\TRhh{\left(PXP\otimes P^{\perp}XP^{\perp} \right)\Gamma}\\
\nonumber&\qquad-4\left(\TRhh{\left(P^{\perp}XP^{\perp}\otimes 
      P^{\perp}XP^{\perp}\right)\Gamma}\right)^{\frac{1}{2}}\left(\TRhh{\left(PXP\otimes 
      P^{\perp}XP^{\perp}\right)\Gamma}\right)^{\frac{1}{2}}.
\end{align*}
As 
$x^2-4bx\geq -\,4b^2$, for $x:=\left(\TRhh{\left(P^{\perp}XP^{\perp}\otimes 
      P^{\perp}XP^{\perp} \right)\Gamma}\right)^{\frac{1}{2}}$ and $b:=\left(\TRhh{\left(PXP\otimes 
      P^{\perp}XP^{\perp}\right)\Gamma}\right)^{\frac{1}{2}}$, one then easily concludes
\begin{align*}
\nonumber T_{\mathrm{R}_1}&\geq -\,4\,\TRhh{\left(PXP\otimes P^{\perp}XP^{\perp}\right)\Gamma}+2\,\TRhh{\left(PXP\otimes 
      P^{\perp}XP^{\perp}\right)\Gamma}\\
&=-\,2\,\TRhh{\left(PXP\otimes P^{\perp}XP^{\perp}\right)\Gamma}.
\end{align*}
The proof is completed by using (\ref{Fundsatz}):
\begin{align*}
 T_{\mathrm{R}_1}\geq-\,2\,\TRhh{\left(PXP\otimes P^{\perp}XP^{\perp}\right)\Gamma}\geq
	-\,8\,\TRh{X\gamma}\TRh{X\left(\gamma-\gamma^2\right)}.
\end{align*}
\end{proof}

The estimate on $T_{\mathrm{R}_2}:=-\,2\,\TRhh{\left(P^{\perp}XP\otimes PXP^{\perp}\right)\Gamma}$ is addressed in 
the next lemma.

\begin{lem}\label{lem4c}
Let $X$, and $P$ and $P^{\perp}$ be as defined in (\ref{defX}) and (\ref{defP}), respectively. 
Assume that $\left(\gamma,\Gamma\right)$ is admissible and fulfills the G- and P-Conditions. Then
\begin{align}
&\quad T_{\mathrm{R}_2}= -\,2\,\TRhh{\left(P^{\perp}XP\otimes PXP^{\perp}\right)\Gamma}\geq-\,8\,\TRh{X\gamma}
      \TRh{X\left(\gamma-\gamma^2\right)}. \label{RemR}
\end{align}
\end{lem}
\begin{proof}
First, the left side is estimated by its absolute value. Then the Cauchy--Schwarz inequality (\ref{CSIN}) is used:
\begin{align}
\nonumber &-2\,\TRhh{\left(P^{\perp}XP\otimes PXP^{\perp}\right)\Gamma}\\
\nonumber &\qquad\geq-\,2\left|\TRhh{\left(P^{\perp}XP\otimes PXP^{\perp}\right)\Gamma}\right|\\
\nonumber &\qquad\geq-\,2\left(\TRhh{\left(P^{\perp}XP^{\perp}\otimes PXP\right)\Gamma}\right)^{\frac{1}{2}}
				\left(\TRhh{\left(PXP\otimes P^{\perp}XP^{\perp}\right)\Gamma}\right)^{\frac{1}{2}}\\
&\qquad=-\,2\,\TRh{\left(PXP\otimes P^\perp X P^\perp\right)\Gamma}. \label{IneqN}
\end{align}
The assertion (\ref{RemR}) follows again from (\ref{Fundsatz}).
\end{proof}

Summing up the results, we obtain the following estimate of the Remainder directly from Lemmata \ref{lem4b} and \ref{lem4c}.

\begin{thm}\label{mthm4b}
Let $X$, and $P$ and $P^{\perp}$ be as defined in (\ref{defX}) and (\ref{defP}), respectively. 
Assume that $\left(\gamma,\Gamma\right)$ is admissible and fulfills the G- and P-Conditions. Then
\begin{align*}
&\qquad T_{\mathrm{R}}\geq -\,16\,\TRh{X\gamma}\TRh{X\left(\gamma-\gamma^2\right)}.
\end{align*}
\end{thm}


\subsection{Estimation of the Main Error Term}

The main error in Theorem \ref{mthm4} results from estimating $T_{\mathrm{MET}}$ in (\ref{MP+R}),
\begin{align}
T_{\mathrm{MET}}=2\,\mathfrak{Re}\left\{\TRhh{\left(P^{\perp}XP\otimes P^{\perp}X P\right)\Gamma}\right\}. \label{TMETeq}
\end{align}
A key observation is that terms of the form $A\otimes B\in\cB(\fh\otimes\fh)$ or 
$\mathrm{Ex}\left(A\otimes B\right)\in\cB(\fh\otimes\fh)$ can be added to $\Gamma$ in (\ref{TMETeq})
without changing the value of $T_{\mathrm{MET}}$, provided $A$
and $B$ commute with $P^\perp\ \text{and}\ P$. 
Therefore, one can also consider
\begin{align*}
T_{\mathrm{MET}}=2\,\mathfrak{Re}\left\{ \TRhh{\left(P^{\perp}XP\otimes P^{\perp}X P\right)\left(
	\Gamma+\mathrm{Ex}\left(\gamma\otimes\mathbbm{1}\right)\right)}\right\}.
\end{align*}
This expression can now be estimated using a variant of the Cauchy--Schwarz inequality given
in the next lemma.

\begin{lem}
Let $X$, and $P$ and $P^{\perp}$ be as defined in (\ref{defX}) and (\ref{defP}), respectively.
Assume that $\left(\gamma,\Gamma\right)$ is admissible and fulfills the G-Condition. Then
\begin{align}
\nonumber& \mathfrak{Re}\left\{\TRhh{\left(P^{\perp}XP\otimes P^{\perp}X P\right)\left(
	\Gamma+\mathrm{Ex}\left(\gamma\otimes\mathbbm{1}\right)\right)}\right\}\\
\nonumber&\qquad\qquad\qquad \geq - \left(\TRhh{\left(PXP^{\perp}\otimes P^{\perp}X P\right)\left(
	\Gamma+\mathrm{Ex}\left(\gamma\otimes\mathbbm{1}\right)\right)}
\right)^{\frac{1}{2}}\\
&\qquad\qquad\qquad\qquad\quad\times\left(\TRhh{\left(P^{\perp}XP\otimes PX P^{\perp}\right)\left(
	\Gamma+\mathrm{Ex}\left(\gamma\otimes\mathbbm{1}\right)\right)} 
\right)^{\frac{1}{2}}.  \label{HiUgl}
\end{align}
\end{lem}

\begin{proof}
 We define 
\begin{align*}
 \left(A,B\right):=\TRhh{\left(A^*\otimes B\right)\left(\Gamma+\mathrm{Ex}\left(\gamma\otimes\mathbbm{1}\right)\right)}
\end{align*}
 on $\cB(\fh)\times\cB(\fh)$ and observe that, because
\begin{align*}
 \TRhh{\left(B^*\otimes A\right)\Gamma}=\TRhh{\left(A\otimes B^*\right)\Gamma}=\overline{\TRhh{\left(A^*\otimes B\right)\Gamma}}
\end{align*}
and
\begin{align*}
 \TRhh{\left(B^*\otimes A\right)\mathrm{Ex}\left(\gamma\otimes\mathbbm{1}\right)}&=\TRh{B^*A\,\gamma}=
  \overline{\TRh{A^*B\,\gamma}}\\
  &=\overline{\TRhh{\left(A^*\otimes B\right)\mathrm{Ex}\left(\gamma\otimes\mathbbm{1}\right)}},
\end{align*}
$\left(\cdot\,,\cdot\right)$ defines a Hermitian form on $\cB(\fh)\times\cB(\fh)$. Furthermore, $\left(\cdot\,,\cdot\right)$
is positive semidefinite, since
\begin{align*}
 \nonumber\left(A,A\right)&=\TRhh{\left(A^*\otimes A\right)\left(\Gamma+\mathrm{Ex}\left(\gamma\otimes\mathbbm{1}\right)\right)}\\
  &\geq\left|\TRh{A\gamma}\right|^2\geq 0
\end{align*}
thanks to the G-Condition. Hence, the Cauchy--Schwarz inequality 
$\left|\left(A,B\right)\right|\leq\left(A,A\right)^{\frac{1}{2}}\left(B,B\right)^{\frac{1}{2}}$ holds true. Applying 
this with $A^*:=B:=P^\perp X P$, we arrive at the asserted estimate (\ref{HiUgl}):
\begin{align*}
 \mathfrak{Re}&\left\{\TRhh{\left(P^{\perp}XP\otimes P^{\perp}X P\right)\left(
	\Gamma+\mathrm{Ex}\left(\gamma\otimes\mathbbm{1}\right)\right)}\right\}\\
  &\qquad=\mathfrak{Re}\left\{\left(A^*,A\right)\right\}\geq-\left|\left(A^*,A\right)\right|\geq
    -\left(A^*,A^*\right)^\frac{1}{2}\left(A,A\right)^\frac{1}{2}\\
  &\qquad=- \left(\TRhh{\left(PXP^{\perp}\otimes P^{\perp}X P\right)\left(
	\Gamma+\mathrm{Ex}\left(\gamma\otimes\mathbbm{1}\right)\right)}
\right)^{\frac{1}{2}}\\
&\qquad\qquad\qquad\qquad\quad\times\left(\TRhh{\left(P^{\perp}XP\otimes PX P^{\perp}\right)\left(
	\Gamma+\mathrm{Ex}\left(\gamma\otimes\mathbbm{1}\right)\right)}
\right)^{\frac{1}{2}}.
\end{align*}

\end{proof}

Now the Main Error Term can be estimated.

\begin{thm}\label{mthm4c}
Let $X$, and $P$ and $P^{\perp}$ be as defined in (\ref{defX}) and (\ref{defP}), respectively. 
Assume that $\left(\gamma,\Gamma\right)$ is admissible and fulfills the G- and P-Conditions. Then
\begin{align*}
T_{\mathrm{MET}}\geq-\,2\,\TRh{X\gamma}\Big[8\,\TRh{X\left(\gamma-\gamma^2\right)}
      \Big(1+4\,\TRh{X\left(\gamma-\gamma^2\right)}\Big)\Big]^{\frac{1}{2}}.
\end{align*}
\end{thm}

\begin{proof}
We rewrite $T_{\mathrm{MET}}$ adding the necessary exchange term to allow for an application of (\ref{HiUgl}):
\begin{align*}
\nonumber T_{\mathrm{MET}}&=2\,\mathfrak{Re}\left\{\TRhh{\left(P^{\perp}XP\otimes P^{\perp}X P\right)\Gamma}\right\}\\
\nonumber &= 2\,\mathfrak{Re}\left\{\TRhh{\left(P^{\perp}XP\otimes P^{\perp}X P\right)
		\left(\Gamma+\mathrm{Ex}\left(\gamma\otimes\mathbbm{1}\right)\right)}\right\}\\
\nonumber	& \geq-\,2\left(\TRhh{\left(PXP^{\perp}\otimes P^{\perp}X P\right)\left(
	\Gamma+\mathrm{Ex}\left(\gamma\otimes\mathbbm{1}\right)\right)}
\right)^{\frac{1}{2}}\\
&\qquad\qquad\qquad\times\left(\TRhh{\left(P^{\perp}XP\otimes PX P^{\perp}\right)\left(
	\Gamma+\mathrm{Ex}\left(\gamma\otimes\mathbbm{1}\right)\right)}
\right)^{\frac{1}{2}}.	
\end{align*}
Note that $\TRhh{\left(PXP^{\perp}\otimes P^{\perp}X P\right)\Gamma}=\TRhh{\left(P^{\perp}XP\otimes PX P^{\perp}\right)\Gamma}$, which 
was already estimated in (\ref{IneqN}). Together with (\ref{Fundsatz}) we obtain
\begin{align*}
 \TRhh{\left(PXP^{\perp}\otimes P^{\perp}X P\right)\Gamma}&\leq \TRh{\left(PXP\otimes P^\perp X P^\perp\right)\Gamma}\\
  &\leq 4\TRh{X\gamma}\TRh{X\left(\gamma-\gamma^2\right)}.
\end{align*}
The two exchange terms have to be treated separately. Using $P^{\perp}, P\leq\mathbbm{1}$, we find
\begin{align*}
 \TRhh{\left(PXP^{\perp}\otimes P^{\perp}X P\right)\mathrm{Ex}\left(\gamma\otimes\mathbbm{1}\right)}
      & = \TRh{P^{\perp} X P \gamma P X P^{\perp}}\\
& = \TRh{P^{\perp} X\sqrt{\gamma} P \sqrt{\gamma} X P^{\perp}}\\
      & \leq \TRh{X\gamma},
\end{align*}
and 
\begin{align*}
  \TRhh{\left(P^{\perp}XP\otimes PX P^{\perp}\right)\mathrm{Ex}\left(\gamma\otimes\mathbbm{1}\right)}&=
      \nonumber \TRh{PXP^{\perp}\gamma P^{\perp}X P}\\
    &=\TRh{X P X X P^{\perp}\gamma P^{\perp}X}\\
\nonumber				&\leq\TRh{ X P X}\TRh{ XP^{\perp}\gamma P^{\perp}X}\\
						&\leq 4\,\TRh{X\gamma}\TRh{X\left(\gamma-\gamma^2\right)},
\end{align*}
where cyclic permutation in the argument of the trace is used together with $ X=X^2 $, 
$ \left(P^\perp\right)^2=P^\perp$ and $P^2=P$ to expand the argument of the trace. 
Then one can use $ X P X\geq 0 $ and $ XP^{\perp}\gamma P^{\perp}X\geq 0$ to estimate by $ X P X \leq \mathbbm{1}\,\TRh{X P X} 
 $. Afterwards, (\ref{Proj2}) can be used. Merging the results, we arrive at the assertion.
\end{proof}


\subsection{Estimation of the Main Part}

In this section it is shown that the remaining terms of (\ref{MP+R}),
\begin{align*}
\nonumber T_{\mathrm{MP}}&=\TRhh{ \left(P X P \otimes P X P\right) \Gamma } +4\,\mathfrak{Re}\left\{\TRhh{\left(PX P^{\perp} 
      \otimes P X P\right)\Gamma}\right\}\\
&\qquad +4\,\TRhh{\left(PXP^{\perp}\otimes P^{\perp}XP\right)\Gamma}, 
\end{align*}
are large enough to cover the HF-part 
$\TRhh{\left(X\otimes X\right)\left(\mathbbm{1}-\mathrm{Ex}\right)\left(\gamma\otimes\gamma\right)}$ in (\ref{Spur1}). Due to this, we call this 
terms the Main Part. As mentioned, the Main Part was extended by an additional term. This extension allows for the
following observation.

\begin{lem}
Let $X$, and $P$ and $P^{\perp}$ be as defined in (\ref{defX}) and (\ref{defP}), respectively. Then
\begin{align}
T_{\mathrm{MP}}=\TRhh{\left(PX\left(P+2 P^{\perp}\right)\otimes\left(P+2 P^{\perp}\right)XP\right)\Gamma}. \label{TmpBP}
\end{align}
\end{lem}

\begin{proof}
Expanding the parentheses on the right side leads to the assertion by using $\TRhh{\left(A\otimes B\right)\Gamma}
  =\TRhh{\left(B\otimes A\right)\Gamma}$. 
\end{proof}

For $A:=\left(P+2 P^{\perp}\right)XP$ we have $T_{\mathrm{MP}}=\TRhh{\left(A^*\otimes A\right)\Gamma}$. 
This provides the use of the G-Condition.

\begin{thm}\label{mthm4d}
Let $X$, and $P$ and $P^{\perp}$ be as defined in (\ref{defX}) and (\ref{defP}), respectively.
Assume that $\left(\gamma,\Gamma\right)$ is admissible and fulfills the G-Condition. Then
\begin{align}
&T_{\mathrm{MP}}-\TRhh{\left(X\otimes X\right)\left(\mathbbm{1}-\mathrm{Ex}\right)
    \left(\gamma\otimes\gamma\right)}\geq -\,22\, \TRh{X\gamma}\TRh{X\left(\gamma-\gamma^2\right)}. \label{MPabs}
\end{align}
\end{thm}
\begin{proof}
The proof is split into two parts. In the first part, the trace of the HF-part is calculated. In the second part, 
the Main Part is estimated by applying the G-Condition with $A:=\left(P+2 P^{\perp}\right)XP$. 
\begin{itemize}
\item[a)] As in (\ref{abs1}), the trace of the HF part can be written as
	\begin{align*}
	\TRhh{\left(X\otimes X\right)\left(\mathbbm{1}-\mathrm{Ex}\right)\left(\gamma\otimes
	\gamma\right)}=\left(\TRh{X\gamma}\right)^2-\TRh{X\gamma X\gamma}.
	\end{align*}
\item[b)] Owing to (\ref{TmpBP}), the G-Condition can be applied directly on the Main Part:
	\begin{align*}
\nonumber	T_{\mathrm{MP}}&=\TRhh{\left(PX\left(P+2 P^{\perp}\right)\otimes
			\left(P+2 P^{\perp}\right)XP\right)\Gamma}\\
\nonumber &\geq	\left|\TRh{PX\left(P+2 P^{\perp}\right)\gamma}\right|^2\\
\nonumber &\qquad\qquad\qquad\qquad	-\TRh{PX\left(P+2 P^{\perp}\right)\gamma\left(P+2 P^{\perp}\right)XP}.\\
\intertext{Due to cyclical permutation, $\left[\gamma,P\right]=\left[\gamma,P^\perp\right]=0$ and $P^\perp P=P P^\perp=0$, some 
traces vanish. The result is}
\nonumber T_{\mathrm{MP}}&\geq \left|\TRh{PX\gamma}\right|^2-\TRh{XPXP\gamma}-4\,\TRh{XPXP^{\perp}\gamma}\\
 &\geq \left(\TRh{PX\gamma}\right)^2-\TRh{PX\gamma X}-4\,\TRh{XP}\TRh{XP^{\perp}\gamma}.
	\end{align*}
In $\TRh{XPXP\gamma}$, $\left[P,\gamma\right]=0$ and $P\leq \mathbbm{1}$ is used to 
write $\TRh{XPXP\gamma}=\TRh{PX\sqrt{\gamma}P\sqrt{\gamma}XP}\leq\TRh{PX\gamma X}$. 
In the last trace, $XP\leq\mathbbm{1}\,\TRh{XP}$ is used. This is possible since 
$\TRh{XPXP^{\perp}\gamma}=\TRh{XPX XP^{\perp}\gamma X}$ and 
$XP^{\perp}\gamma X=\left|XP^{\perp}\sqrt{\gamma}\right|^2\geq 0$ together with $XPX\geq 0$.
\end{itemize}
Before adding up the estimates, we note that
\begin{align*}
\TRh{P^{\perp}X\gamma X\gamma}&=\TRh{\sqrt{\gamma} X \sqrt{\gamma}P^{\perp}\sqrt{\gamma}X\sqrt{\gamma}}\geq0
\end{align*}
and
\begin{align}
\nonumber &\left(\TRh{X\gamma}\right)^2-\left(\TRh{PX\gamma}\right)^2\\
\nonumber&\qquad\qquad
=\Big(\TRh{X\gamma}+\TRh{PX\gamma}\Big) \Big(\TRh{X\gamma}-\TRh{PX\gamma}\Big)\\
 &\qquad\qquad
=\Big(\TRh{X\gamma}+\TRh{PX\gamma}\Big)\TRh{P^{\perp}X\gamma}. \label{ugl7}
\end{align}
Furthermore, one has $\TRh{XP}\TRh{XP^{\perp}\gamma}\leq 4\,\TRh{X\gamma}\TRh{X\left(\gamma-\gamma^2\right)}$.
These results can now be applied together with a) and b) to the left side of (\ref{MPabs}):
\begin{align*}
	\nonumber&T_{\mathrm{MP}}-\TRhh{\left(X\otimes X\right)\left(\mathbbm{1}-\mathrm{Ex}\right)
			      \left(\gamma\otimes\gamma\right)}\\
	\nonumber&\qquad \geq-\left(\left(\TRh{X\gamma}\right)^2
		-\left(\TRh{PX\gamma}\right)^2\right)+\TRh{X\gamma X\gamma}-\TRh{PX\gamma X}\\
		\nonumber &\qquad\qquad\qquad\qquad\qquad\qquad\qquad\qquad
				-16\,\TRh{X\gamma}\TRh{X\left(\gamma-\gamma^2\right)}.
\end{align*}
At this point we use (\ref{ugl7}), split $\TRh{X\gamma X\gamma}$ into $\TRh{P X\gamma X\gamma}+\TRh{P^{\perp}X\gamma X\gamma}$,
and rearrange: 
\begin{align*}
	\nonumber&T_{\mathrm{MP}}-\TRhh{\left(X\otimes X\right)\left(\mathbbm{1}-\mathrm{Ex}\right)
			      \left(\gamma\otimes\gamma\right)}\\
	\nonumber &\qquad \geq-\left(\TRh{X\gamma}+\TRh{PX\gamma}\right)\TRh{P^{\perp}X\gamma}\\
	\nonumber &\qquad\qquad\qquad-\TRh{PX\gamma X\left(\mathbbm{1}-\gamma\right)}
			+\TRh{P^{\perp}X\gamma X\gamma}\\
	\nonumber &\qquad\qquad\qquad-16\,\TRh{X\gamma}\TRh{X\left(\gamma-\gamma^2\right)}.
\end{align*}
Then, with $P\leq\mathbbm{1}$, $P^{\perp}\leq 2\left(1-\gamma\right)$, and $\TRh{P^{\perp}X\gamma X\gamma}\geq 0$, one obtains
\begin{align*}
	\nonumber&T_{\mathrm{MP}}-\TRhh{\left(X\otimes X\right)\left(\mathbbm{1}-\mathrm{Ex}\right)
			      \left(\gamma\otimes\gamma\right)}\\
	\nonumber &\qquad \geq -\,4\,\TRh{X\gamma}\TRh{X\left(\gamma-\gamma^2\right)}
		-\TRh{PX\gamma X\left(\mathbbm{1}-\gamma\right)}\\
	\nonumber &\qquad\qquad\qquad\qquad\qquad\qquad\qquad\qquad
				-16\,\TRh{X\gamma}\TRh{X\left(\gamma-\gamma^2\right)}.
\end{align*}
We continue with the inequality $X\gamma X\leq\mathbbm{1}\TRh{X\gamma X}$. This is allowed because
$\TRh{PX\gamma X\left(\mathbbm{1}-\gamma\right)}=\TRh{X\gamma X X\left(\mathbbm{1}-\gamma\right)PX}$ and 
$X\left(\mathbbm{1}-\gamma\right)PX=
  X\sqrt{\mathbbm{1}-\gamma}P\sqrt{\mathbbm{1}-\gamma}X=
\left|X\sqrt{\mathbbm{1}-\gamma}P\right|^2\geq 0$ together with $X\gamma X\geq 0$:
\begin{align*}
	\nonumber &T_{\mathrm{MP}}-\TRhh{\left(X\otimes X\right)\left(\mathbbm{1}-\mathrm{Ex}\right)
			      \left(\gamma\otimes\gamma\right)}\\
	\nonumber &\qquad \geq -\,20\,\TRh{X\gamma}\TRh{X\left(\gamma-\gamma^2\right)}
			-\TRh{X\gamma X}\TRh{X\left(\mathbbm{1}-\gamma\right)P X}\\
	\nonumber &\qquad \geq -\,20\,\TRh{X\gamma}\TRh{X\left(\gamma-\gamma^2\right)}
				-2\,\TRh{X\gamma}\TRh{X\left(\gamma-\gamma^2\right)}\\
	& \qquad = -\,22\,\TRh{X\gamma}\TRh{X\left(\gamma-\gamma^2\right)}.
\end{align*}
The last inequality follows from $P\leq2\gamma$.
\end{proof}

Finally, the proof of Theo\-rem \ref{mthm4} is completed by the estimation of $T_{\mathrm{R}}$ in Theo\-rem \ref{mthm4b}, $T_{\mathrm{MET}}$ in Theo\-rem \ref{mthm4c} and 
$T_{\mathrm{MP}}-\TRhh{\left(X\otimes X\right)\left(\mathbbm{1}-\mathrm{Ex}\right)\left(\gamma\otimes\gamma\right)}$ in
Theo\-rem \ref{mthm4d}. In each of this theorems, the G-Condition was used to generate bounds. The P-Condition was only applied to
provide the use of the Cauchy--Schwarz inequality. In the end, it is remarkable that the Q-Condition is not needed for the proof of the correlation estimate.


\section{Summary}

We have obtained several results in the last section, which were merged in the main theorem, Theorem \ref{mthm4}:
\begin{align*}
\TRhh{\left(X\otimes X\right)\left(\Gamma-\left(1-\text{Ex}\right)
		\left(\gamma\otimes\gamma\right)\right)}\geq-\,\TRh{X\gamma},
\end{align*}
\begin{align*}
\qquad T_{\text{R}}\geq - \,16\,\TRh{X\gamma}\TRh{X\left(\gamma-\gamma^2\right)},
\end{align*}
\begin{align*} T_{\text{MET}}\geq -\,2\,\TRh{X\gamma}\Big[8\,\TRh{X\left(\gamma-\gamma^2\right)}
	      \Big(1+4\,\TRh{X\left(\gamma-\gamma^2\right)}\Big)\Big]^\frac{1}{2},
\end{align*}
\begin{align*}
T_{\text{MP}}-\TRhh{\left(X\otimes X\right)\left(\mathbbm{1}-\text{Ex}\right)
	      \left(\gamma\otimes\gamma\right)}\geq -\,22\, \TRh{X\gamma}\TRh{X\left(\gamma-\gamma^2\right)}.
\end{align*}

Denoting $b:=\TRh{X\gamma}$ and $a:=\sqrt{\TRh{X\left(\gamma-\gamma^2\right)}}$, one can rewrite the estimates
for $\TRhh{\left(X\otimes X\right)\Gamma^{\left(\mathrm{T}\right)}}=
\TRhh{\left(X\otimes X\right)\left(\Gamma-\left(1-\mathrm{Ex}\right)
		\left(\gamma\otimes\gamma\right)\right)}$ as follows:
\begin{align}
\TRhh{\left(X\otimes X\right)\Gamma^{\left(\mathrm{T}\right)}}\geq	
		-\,b\,\min\left\{1;\ a\left(38a+2\sqrt{8+32a^2}\right)\right\}. \label{CorrUGL}
\end{align}

A suitable choice of $a\leq b$ in (\ref{CorrUGL}) leads to the following correlation estimation.

\begin{thm} \label{CorrelationIneq}
Let $X$, and $P$ and $P^{\perp}$ be as defined in (\ref{defX}) and (\ref{defP}), respectively. 
Assume that $\left(\gamma,\Gamma\right)$ is admissible and fulfills the G- and P-Conditions. Then
 \begin{align*}
\TRhh{\left(X\otimes X\right)\Gamma^{\left(\mathrm{T}\right)}} \geq	
		-\,\TRh{X\gamma}\,\min\left\{1;\ 10\sqrt{\TRh{X\left(
		\gamma-\gamma^2\right)}}\right\}.
\end{align*}
\end{thm}
\begin{proof} The minimum in (\ref{CorrUGL}) is $a\left(38a+2\sqrt{8+32a^2}\right)$ for $0< a\leq \frac{1}{\sqrt{94}}$ and,
thus,  $\frac{1}{a}\geq \left(38a+2\sqrt{8+32a^2}\right)$.
 Since $\left(38a+2\sqrt{8+32a^2}\right)$ is monotonously increasing
in $a$, we find
\begin{align}
\left(38a+2\sqrt{8+32a^2}\right)\leq \sqrt{94}< 10 \label{const},
\end{align}
which implies the assertion.
\end{proof}

\begin{rem}
In section 4.1, we have split the eigenvalues of $\gamma$ in eigenvalues which are larger than
$\frac{1}{2}$ and lower or equal $\frac{1}{2}$. In fact, this split turns out to be almost optimal 
and (\ref{const}) cannot be sharpened by another choice of $P$ and $P^\perp$. 
\end{rem}

Up to the constant (\ref{const}), 
Theorem \ref{CorrelationIneq} is exactly the result which was already obtained
in \cite{VBA}. The difference of the constants comes, on the one hand, from a different arrangement of the terms
of $\TRhh{\left(X\otimes X\right)\Gamma^{\left(\mathrm{T}\right)}}$ and, on the other hand, from the fact that
in \cite{VBA} also the Q-Condition was used, which can be seen implicitly in estimate (68) in \cite{VBA}.  
With the result of Theorem \ref{CorrelationIneq} we can immediately perform the 
integration in the Feffermann--de la Llave identity according to \cite{VBA} which leads to an estimate of
$\TRhh{V\left(\Gamma-(1-\mathrm{Ex})(\gamma\otimes\gamma)\right)}$.


\bibliography{BKM_REF}
\Addresses
\end{document}